\documentclass[sn-mathphys,Numbered,iicol]{sn-jnl}% Math and Physical Sciences Reference Style
\usepackage[utf8]{inputenc}
\usepackage{chngcntr}
\usepackage{graphicx}%
\usepackage{amsmath,amssymb,amsfonts}%
\usepackage{amsthm}%
\usepackage{manyfoot}% required by sn-jnl footnote hooks
\usepackage[title]{appendix}%
\usepackage{xcolor}%
\usepackage{textcomp}%
\usepackage{booktabs}%
\usepackage{tabularx}%
\usepackage{caption}%
\counterwithin{equation}{section}
\theoremstyle{thmstyleone}%
\newtheorem{theorem}{Theorem}%  meant for continuous numbers
\newtheorem{proposition}[theorem]{Proposition}% 

\theoremstyle{thmstyletwo}%

\theoremstyle{thmstylethree}%

\begin{document}

\title[Dirac--Bergmann analysis of SW-mapped NC $U(1)$ electrodynamics]{Dirac--Bergmann analysis of SW-mapped non-commutative $U(1)$ electrodynamics with external currents}

\author[1]{\fnm{A. G.} \sur{Andarcia-Caballero}}\email{alejandroandarcia@hotmail.com}

\author*[1]{\fnm{J.} \sur{Manuel-Cabrera}}\email{jaime.manuel@ujat.mx}

\author*[1]{\fnm{J. M.} \sur{Paulin-Fuentes}}\email{jorge.paulin@ujat.mx}

\equalcont{These authors contributed equally to this work.}

\affil*[1]{\orgdiv{Divisi\'on Acad\'emica de Ciencias B\'asicas}, \orgname{Universidad Ju\'arez Aut\'onoma de Tabasco}, \orgaddress{\street{Km 1 Carretera
Cunduac\'an-Jalpa}, \city{ Cunduac\'an}, \postcode{86690}, \state{Tabasco}, \country{ M\'exico}}}

\abstract{Non-commutative electrodynamics obtained through the Seiberg--Witten map ceases to have equivalent action-level and equation-level realizations once fixed external currents are introduced, and in the action-level construction associated with the Banerjee current map the canonical location of this source-induced obstruction has remained unclear. Working in the full phase space and treating the current as prescribed and non-dynamical, we apply the Dirac--Bergmann algorithm without imposing current conservation as an external condition. The preservation of the Gauss-type secondary constraint produces a third-stage candidate whose phase-space expression is shown to be algebraically identical, at first order in the non-commutativity parameter and for purely space--space non-commutativity, to the canonical pullback of the divergence of the mapped Euler--Lagrange equations. This identity locates the source-compatibility obstruction directly within the Dirac chain. For generic inhomogeneous sources, the next consistency step feeds this object back into the primary multiplier through a source-dependent kernel, so the chain closes by multiplier fixing rather than by the generic appearance of a quaternary constraint. Reduced-phase-space results, including the gauge generator, Dirac brackets and degree-of-freedom count, are obtained only in a restricted sufficient first-class subcase; no broader claim is made for arbitrary source profiles.}


\maketitle

%==============================================================================

\section{Introduction}
\label{sec:intro}
%==============================================================================

Non-commutative (NC) field theory deforms the coordinate algebra according to
\begin{equation}
  \bigl[x^{\mu},x^{\nu}\bigr] = i\,\theta^{\mu\nu},
  \label{eq:NCcommutator}
\end{equation}
where $\theta^{\mu\nu}$ is a constant antisymmetric tensor
\cite{Heisenberg1938,Snyder1947}. Through the Seiberg--Witten (SW) map, NC
gauge theories can be rewritten order by order in $\theta^{\mu\nu}$ in terms
of ordinary fields and standard spacetime coordinates
\cite{Douglas2001,Szabo2003,Seiberg1999,Bichl2002,Riad2000,Madore2000,Fidanza2002}.
In the absence of sources, the first-order map leads to the familiar
gauge-invariant weakly nonlinear deformation of Maxwell theory
\cite{Berrino2003,Carroll2001}.

With external sources the situation is subtler. In ordinary Maxwell theory a
prescribed current is compatible with gauge symmetry only if it is conserved.
In the NC setting, Adorno \emph{et al.}~\cite{Adorno2011CNE} showed that the
action-level and equation-level implementations of the SW map are no longer
equivalent in the presence of fixed external currents: mapping after deriving
the equations preserves gauge covariance, whereas mapping the action before
variation does not. The Hamiltonian meaning of that mismatch, and in
particular the point at which it enters the Dirac chain, has remained less
clear.

Related external-source analyses show that the nature of the source itself is
part of the problem rather than a secondary modeling choice. Cabo and
Shabad treat a non-Abelian classical current under the restriction
$J_4^{\,a}=0$, so that the source modifies the canonical consistency
conditions without producing the same temporal shift of Gauss's law examined
here~\cite{CaboShabad}. Przeszowski's canonical path-integral analysis of
Yang--Mills theory with arbitrary external sources goes further in a different
direction: the naive source coupling is replaced by a modified source sector
before quantization, precisely because the canonical structure is sensitive to
the way the external current is introduced~\cite{Przeszowski1988}. Sikivie and
Weiss, by contrast, treat static external sources as fixed backgrounds used to
define distinct classical solution sectors, such as Coulomb and screening
branches, rather than to analyze first-/second-class structure in the sense of
Dirac--Bergmann~\cite{SikivieWeiss1978}. In NC Maxwell theory, Kruglov adds a
prescribed Abelian current directly to the first-order SW-expanded
Lagrangian and studies the resulting nonlinear field equations and
energy--momentum tensors, but without a Hamiltonian source-sector analysis
of the type carried out here~\cite{Kruglov2002}. These comparisons delimit the
present problem more sharply: the issue is not merely that an external source
is present, but how a fixed four-current enters the canonical consistency
chain of the action-level SW-mapped theory.

To date, no Dirac--Bergmann analysis of the first-order action-level
SW-mapped theory with fixed external currents has been carried out: the
existing treatments either operate exclusively at the Lagrangian
level~\cite{Adorno2011CNE,Kruglov2002}, or address external sources in
canonical frameworks that do not involve the SW
map~\cite{CaboShabad,Przeszowski1988}. As a consequence, the precise point in
the canonical consistency chain at which the source-induced obstruction
enters---and the way in which it depends on the choice of SW current
map---has remained undetermined.

This paper fills that gap for the first-order action-level realization based
on the Banerjee current map~\cite{Banerjee2004PRD}. We treat $J^\mu$ as a
prescribed non-dynamical source and apply the Dirac--Bergmann algorithm in the
full phase space, keeping $A_0$ and $\pi^0$ as canonical variables throughout.
The analysis is restricted to first order in $\theta$ and to purely spatial
non-commutativity. Within that framework, the secondary Gauss-type constraint
already shows where the source enters the canonical structure, while its
preservation generates a third-stage source-dependent object whose canonical
form is shown independently to be algebraically identical to the canonical
pullback of the divergence of the Banerjee Euler--Lagrange equations. Because
this identity is tied to the specific $F_{\alpha\beta}$-dependent terms of the
Banerjee map, it is not a generic consequence of any action-level SW
realization; rather, it characterizes the map dependence of the obstruction at
the canonical level. The next consistency step closes the generic source sector
by multiplier fixing---itself a positive structural result, not merely a
limitation---whereas an exact reduced Hamiltonian description exists only in a
restricted sufficient first-class subcase.

The rest of the paper is organized as follows.
Section~\ref{sec:NCMaxwell} fixes the working model: the conventions
(Sec.~\ref{sec:conventions}), the source-free SW-mapped action in closed
commutative form (Sec.~\ref{sec:SourceFreeLagrangian}), and the current-sector
choice---the Banerjee map---that differentiates the present construction from
that of Adorno \emph{et al.} (Sec.~\ref{sec:SourceFreeEquivalence}).
Section~\ref{sec:DiracBergmann} then develops the Dirac--Bergmann analysis of
the unrestricted external-current problem, up to and including the tertiary
stage, and establishes the central identity of the paper
(Proposition~\ref{prop:DivELIdentity}): the tertiary-stage candidate coincides,
in the canonical framework, with the divergence of the mapped
Euler--Lagrange equations. Section~\ref{sec:ReducedDescription} develops the
reduced phase-space description that becomes available only in a restricted
sufficient first-class subcase, and uses it to compare the Banerjee reduction
with the gauge-invariant equation-level formulation of
Ref.~\cite{Adorno2011CNE}. The structural consequences, scope, and limitations
are analyzed in Sec.~\ref{sec:Discussion}, and the findings are summarized in
Sec.~\ref{sec:Conclusions}.

\section{Working model and current map}
\label{sec:NCMaxwell}
%==============================================================================

\subsection{Conventions and notation}
\label{sec:conventions}

We first fix the notation needed in the rest of the paper. We work in Minkowski spacetime with metric signature 
$\eta^{\mu\nu}=\mathrm{diag}(+1,-1,-1,-1)$ and set $\hbar=c=1$. Greek indices 
$\mu,\nu,\rho,\ldots$ run over $\{0,1,2,3\}$, while Latin indices 
$i,j,k,\ldots$ are restricted to spatial components $\{1,2,3\}$. The 
non-commutativity parameter $\theta^{\mu\nu}=-\theta^{\nu\mu}$ is a constant 
antisymmetric tensor with dimensions $[\text{length}]^2$.

To avoid issues with unitarity and acausal propagation, we consider purely 
space-space non-commutativity and set $\theta^{0i}=0$ identically. From the 
spatial components $\theta^{ij}$ we form the NC vector
\begin{equation}
\theta^i \equiv \frac{1}{2}\,\varepsilon^{ijk}\,\theta_{jk}\,,
\label{eq:thetaVector}
\end{equation}
where $\varepsilon^{ijk}$ is the Levi-Civita symbol, normalized by 
$\varepsilon^{123}=+1$. The commutative electromagnetic field strength is 
$F_{\mu\nu} = \partial_\mu A_\nu - \partial_\nu A_\mu$, and we decompose it 
into electric and magnetic components as
\begin{equation}
E^i = -F^{0i}\,, \quad B^i = \frac{1}{2}\,\varepsilon^{ijk} F_{jk}\,.
\label{eq:EB}
\end{equation}

\subsection{Source-free mapped action}
\label{sec:SourceFreeLagrangian}

We begin with the source-free NC Maxwell action, since the source sector will later be added on top of the same mapped gauge block \cite{Banerjee2004PRD,Abe2003}
\begin{equation}
S_{\text{NC}}
=
\int d^4x\,
\left(-\frac{1}{4}\,\hat{F}_{\mu\nu}\star\hat{F}^{\mu\nu}\right),
\label{eq:NCAction}
\end{equation}
where the NC field strength is given by\footnote{In order to construct the non-commutative version of a given field theory, we replace the ordinary pointwise product of fields in the action by the Moyal $\star$-product. The only identities needed below are the standard first-order expansion, cyclicity of the integral under suitable boundary conditions, and the Leibniz rule for ordinary derivatives.}
\begin{equation}
\begin{aligned}
\hat{F}_{\mu\nu}
={}&\, \partial_\mu \hat{A}_\nu
- \partial_\nu \hat{A}_\mu \\
&\, - i[\hat{A}_\mu, \hat{A}_\nu]_\star\,.
\end{aligned}
\label{eq:NCFieldStrength}
\end{equation}

For NC $U(1)$, the ordinary commutator vanishes, but the $\star$-commutator yields derivative corrections starting at $O(\theta)$; these are consistently captured in the SW expansion.
The Seiberg--Witten map expresses the NC quantities $(\hat{A}_\mu,\hat{F}_{\mu\nu})$ 
in terms of their commutative counterparts $(A_\mu,F_{\mu\nu})$ as a power series 
in~$\theta$~\cite{Seiberg1999,Bichl2002}. To first order one finds
\begin{equation}
\begin{aligned}
\hat{A}_\mu
={}&\, A_\mu
- \frac{1}{2}\,\theta^{\alpha\beta}\,A_\alpha
\left(\partial_\beta A_\mu + F_{\beta\mu}\right) \\
&\, + \mathcal{O}(\theta^2)\,.
\end{aligned}
\label{eq:SWPotential}
\end{equation}
From this, the induced map for the field strength follows:
\begin{equation}
\begin{aligned}
\hat{F}_{\mu\nu}
={}&\, F_{\mu\nu}
+ \theta^{\alpha\beta}
\Bigl(
F_{\mu\alpha}F_{\nu\beta}
- A_\alpha\partial_\beta F_{\mu\nu}
\Bigr) \\
&\, + \mathcal{O}(\theta^2)\,.
\end{aligned}
\label{eq:SWFieldStrength}
\end{equation}
Substituting Eq.~\eqref{eq:SWFieldStrength} into Eq.~\eqref{eq:NCAction} and using the standard cyclicity of the $\star$-product under the integral, we may discard total derivatives. The resulting Lagrangian can be written entirely in terms of the commutative field strength $F_{\mu\nu}$:
\begin{align}
\mathcal{L}_{\text{ws}} =& -\frac{1}{4}\,F_{\mu\nu}F^{\mu\nu} 
+ \frac{1}{8}\,\theta^{\alpha\beta}\,F_{\alpha\beta}\,F_{\mu\nu}F^{\mu\nu} \notag \\
&- \frac{1}{2}\,\theta^{\alpha\beta}\,F_{\mu\alpha}\,F_{\nu\beta}\,F^{\mu\nu}\,.
\label{eq:LagrangianCommutative}
\end{align}

Relative to the standard Maxwell Lagrangian, NC electrodynamics acquires 
terms cubic in $F_{\mu\nu}$, signalling that non-commutativity induces 
non-linearities already at the classical level.

Using Eqs.~\eqref{eq:EB} and~\eqref{eq:thetaVector}, the Lagrangian can be 
rewritten in the more transparent form
\begin{align}
\mathcal{L}_{\text{ws}}
&= \tfrac{1}{2}(E^2 - B^2)\bigl[1 + (\boldsymbol{\theta}\!\cdot\!\mathbf{B})\bigr]
\nonumber\\
&\quad - (\boldsymbol{\theta}\!\cdot\!\mathbf{E})(\mathbf{E}\!\cdot\!\mathbf{B})
+ \mathcal{O}(\theta^2)\,.
\label{eq:LagrangianEB}
\end{align}
This expression makes explicit how NC corrections couple the vector 
$\boldsymbol{\theta}$ to the magnetic field and generate a mixed 
$\mathbf{E}$--$\mathbf{B}$ term that vanishes in the commutative limit 
$\theta\to 0$. As we will see in Sec.~\ref{sec:DiracBergmann}, these terms 
have direct implications for the constraint structure, specifically for the 
form of the secondary (Gauss-type) constraint and the gauge-breaking mechanism 
that emerges when external sources are included.

\subsection{The Banerjee current map and relation to earlier work}
\label{sec:SourceFreeEquivalence}

Before turning to the constraint analysis, we specify the current-sector
choice and its relation to earlier work. In the absence of external sources,
both the Banerjee and the Adorno \emph{et al.}\ constructions reduce to the
same first-order SW-mapped Maxwell theory: with the convention matching
$16\pi c=4$, $g=1$, the source-free mapped action of
Ref.~\cite{Adorno2011CNE} reduces exactly to $\mathcal{L}_{\text{ws}}$ in
Eq.~\eqref{eq:LagrangianCommutative}, with $f_{\mu\nu}\equiv F_{\mu\nu}$. The
source-free sector is therefore common to both formulations; the difference
arises when a fixed external current is included.

For the current sector we adopt the standard first-order Banerjee
map~\cite{Banerjee2004PRD}, corresponding to the parameter choice
$(c_1,c_2,c_3)=(1,\tfrac{1}{2},0)$:
\begin{align}
\hat{J}^\mu
&= J^\mu
- A_\alpha\,\theta^{\alpha\beta}\,\partial_\beta J^\mu
+ F_{\alpha\beta}\,J^\beta\,\theta^{\mu\alpha}
\notag\\
&\quad
+ \frac{1}{2}F_{\alpha\beta}J^\mu\theta^{\alpha\beta}
+ \mathcal{O}(\theta^2)\,.
\label{eq:JMap}
\end{align}
This parameter choice satisfies the SW consistency condition
$\delta_\lambda\hat J^\mu = 0$ for ordinary gauge transformations
$\delta_\lambda A_\mu = \partial_\mu\lambda$, as demonstrated in
Ref.~\cite{Banerjee2004PRD} (see also~\cite{Brace2001} for the general
analysis of SW-map ambiguities in the current sector
and~\cite{BanerjeeGhosh2002} for the implications of these maps in the
anomalous current algebra). The ambiguity at first
order in $\theta$ resides in the relative weight of the
$F_{\alpha\beta}$-dependent terms, which is fixed by the condition that the
mapped current transforms covariantly; the choice
$(c_1,c_2,c_3)=(1,\tfrac{1}{2},0)$ is the unique solution within the
parametrization of Ref.~\cite{Banerjee2004PRD}.

Relative to the minimal current map used by Adorno \emph{et al.},
\begin{equation}
\check{j}^\mu
=
j^\mu + g\,\theta^{\alpha\beta}A_\alpha\partial_\beta j^\mu\,,
\label{eq:AdornoMinimalMap}
\end{equation}
the Banerjee choice~\eqref{eq:JMap} contains additional $F_{\alpha\beta}$-dependent
terms. These terms belong to the first-order SW-map ambiguity of the current
sector and are consistent with the analysis of Ref.~\cite{Adorno2011CNE}.
They leave the source-free theory unchanged but modify the source coupling at
$\mathcal{O}(\theta)$, which is precisely the effect relevant for the
Hamiltonian constraint structure analyzed below.

Adorno \emph{et al.}\ showed that applying the SW map at the action level
before variation yields equations of motion that are not gauge covariant when
$J^\mu$ is prescribed and non-dynamical, whereas applying the map after
variation (at the equation-of-motion level) preserves gauge covariance. They
conclude that the consistent route in the presence of external sources is the
equation-level one. The question addressed in the present work is therefore
\emph{not} whether the SW-mapped action is gauge covariant---that it is not,
as Adorno \emph{et al.}\ established---but whether its constrained
Hamiltonian structure provides a precise canonical diagnosis of that
obstruction, and how the diagnosis depends on the choice of SW current map.
From the viewpoint of constrained dynamics, the Lagrangian mismatch of
Ref.~\cite{Adorno2011CNE} becomes a question about closure, preservation,
and first-/second-class classification in the full
Dirac--Bergmann~\cite{Sundermeyer1982, HenneauxTeitelboim1992} framework,
which is developed in Sec.~\ref{sec:DiracBergmann}.

\section{Dirac--Bergmann analysis in full phase space}
\label{sec:DiracBergmann}
%==============================================================================

With the working Lagrangian $\mathcal{L}=\mathcal{L}_{\rm ws}+\mathcal{L}_J$ and
the Banerjee current map fixed in Sec.~\ref{sec:NCMaxwell}, we now submit the
theory to the Dirac--Bergmann algorithm in the full phase space, treating
$J^\mu$ as prescribed and non-dynamical and \emph{without} imposing
$\partial_\mu J^\mu=0$ from the outside. The goal of this section is twofold.
First, to identify the canonical locus at which the external source first
enters the consistency chain, which will turn out to be the tertiary stage
(Sec.~\ref{sec:MainIdentity}, Proposition~\ref{prop:DivELIdentity}). Second, to
characterize the algebraic structure of that tertiary obstruction through the
three source kernels $\Theta^l$, $\Sigma^{lk}$, $\Xi^l$, which determine how
the algorithm closes and in which restricted sector a standard first-class
reduction remains available (Sec.~\ref{sec:SectorStructure}). Throughout this
section, no restriction is placed on the source profile; the sectorial
discussion emerges from the algorithm itself.

\subsection{Source Lagrangian, primary constraint, and primary Hamiltonian}
\label{sec:CurrentConsistency}

We split the working Lagrangian into a source-free gauge block and a mapped
source block,
\begin{equation}
\mathcal{L}=\mathcal{L}_{\text{ws}}+\mathcal{L}_J\,,
\label{eq:LagrangianSplit}
\end{equation}
where $\mathcal{L}_{\text{ws}}$ is the source-free piece in
Eq.~\eqref{eq:LagrangianEB}, and $\mathcal{L}_J$ denotes the coupling to the
prescribed external current after the SW map. Substituting the Banerjee
current map~\eqref{eq:JMap} into the interaction term and using the standard
properties of the $\star$-product, one obtains
\begin{equation}
\begin{split}
\mathcal{L}_J
&= -A_\mu \Bigl[
J^\mu
+\theta^{\mu\alpha}J^\beta
\Bigl(
\partial_\alpha A_\beta
-\tfrac{1}{2}\partial_\beta A_\alpha
\Bigr)
\Bigr]
\\
&\qquad
+\mathcal{O}(\theta^2)\,.
\end{split}
\label{eq:LagrangianSource}
\end{equation}

With the source Lagrangian in hand, we now proceed to the canonical
formulation. The momenta conjugate to $A_\lambda$ are defined as
\begin{equation}
\pi^\lambda \equiv \frac{\delta \mathcal{L}}{\delta \dot{A}_\lambda}\,,
\label{eq:ConjugateMomenta}
\end{equation}
where the dot denotes $\partial_0$. Since $\mathcal{L}$ does not contain 
$\dot{A}_0$, we immediately obtain the primary constraint
\begin{equation}
\phi_1 \equiv \pi^0 \approx 0\,.
\label{eq:PrimaryConstraint}
\end{equation}
The spatial components of the momentum decompose into source-free and 
source-dependent contributions,
\begin{equation}
\pi^i = [\pi^i]_{\text{ws}} + [\pi^i]_J\,,
\label{eq:MomentaSplit}
\end{equation}
where
\begin{align}
[\pi^i]_{\text{ws}}
&= \bigl[\tfrac{1}{2}F_{kl}\theta^{kl} - 1\bigr]F^{0i}
   + F^{0j}F_{jk}\theta^{ki} \nonumber\\
&\quad + F^{0m}F_{kj}\eta_{lm}\eta^{ji}\theta^{lk}\,,
\label{eq:PiWS} \\
[\pi^i]_J &= \tfrac{1}{2}\theta^{ji}A_j J^0\,.
\label{eq:PiJ}
\end{align}
The combined spatial momentum is therefore
\begin{align}
\pi^i
&= \bigl[\tfrac{1}{2}F_{kl}\theta^{kl} - 1\bigr]F^{0i}
   + F^{0j}F_{jk}\theta^{ki}
\nonumber\\
&\quad + F^{0m}F_{kj}\eta_{lm}\eta^{ji}\theta^{lk}
   + \tfrac{1}{2}\theta^{ji}A_j J^0\,.
\label{eq:Pi_i_final}
\end{align}
It will be useful throughout the paper to isolate the purely field-strength,
source-independent $\mathcal{O}(\theta)$ deformation of the momentum. Writing
\begin{equation}
\pi^i
=
-\,F^{0i}
+\mathcal K^i(F)
+\tfrac{1}{2}\theta^{ji}A_jJ^0,
\label{eq:KofFImplicit}
\end{equation}
the field-strength bilinear kernel is
\begin{equation}
\begin{aligned}
\mathcal K^i(F)
={}&\tfrac{1}{2}\,(\theta^{kl}F_{kl})\,F^{0i}
+\theta^{ki}F^{0j}F_{jk} 
\\ &
+\theta^{lk}\eta_{lm}\eta^{ji}F^{0m}F_{kj}\,.
\end{aligned}
\label{eq:KofFDef}
\end{equation}
Thus $\mathcal{K}^i(F)$ collects all the $\mathcal{O}(\theta)$ field-strength
bilinear corrections that deform the Gauss-type canonical structure without
introducing explicit dependence on the potential beyond the separate source
term $\tfrac{1}{2}\theta^{ji}A_jJ^0$. This decomposition will be used both in
the compact form of the tertiary-stage analysis (Sec.~\ref{sec:MainIdentity})
and in the on-shell Gauss-law comparison (Sec.~\ref{sec:OnShellComparison}).

Solving Eq.~\eqref{eq:Pi_i_final} perturbatively for $F^{0i}$ to first order
in $\theta$ via $M^{i}{}_{j}=\delta^{i}{}_{j}+\mathcal{O}(\theta)$, one finds
\begin{align}
F^{0i}
&= \pi^i\Bigl(\tfrac{1}{2}F_{kl}\theta^{lk} - 1\Bigr)
- \pi^k\bigl(F_{kl}\theta^{li} + F_{ml}\eta^{im}\eta_{nk}\theta^{ln}\bigr)
\nonumber\\
&\quad + [\pi^i]_J\,.
\label{eq:F0i_solution}
\end{align}

The canonical Hamiltonian splits into a source-free and a source-dependent part,
\begin{equation}
H_c = H_{\text{ws}} + H_J = \int \left(\mathcal{H}_{\text{ws}} + \mathcal{H}_J
\right) d^3x\,,
\label{eq:CanonicalHamiltonian}
\end{equation}
with
\begin{align}
\mathcal{H}_{\text{ws}}
&= \pi_i\partial^i A_0
   + \tfrac{1}{4}F_{ij}F^{ij}\Bigl[1 - \tfrac{1}{2}F_{ml}\theta^{ml}\Bigr]
\nonumber\\
&\quad - \tfrac{1}{2}\eta_{00}\eta_{ij}\pi^i\pi^j
   \Bigl[1 + \tfrac{1}{2}F_{ml}\theta^{ml}\Bigr]
\nonumber\\
&\quad + \tfrac{1}{2}F^{ip}F_{li}F_{pm}\theta^{ml}
\nonumber\\
&\quad - F_{li}\eta_{00}\eta_{mj}\pi^i\pi^j\theta^{ml}\,,
\label{eq:H_ws}
\end{align}
and
\begin{align}
\mathcal{H}_J
&= A_0 J^0 + A_j J^j
\nonumber\\
&\quad + \tfrac{1}{2}\theta^{jl}\bigl[
   A_j J^0\partial_l A_0 + A_j J^k\partial_l A_k
\nonumber\\
&\quad\qquad
   + A_j J^0\eta_{00}\eta_{li}\pi^i
   - A_j F_{kl}J^k\bigr]\,.
\label{eq:H_J}
\end{align}
The primary Hamiltonian is then
\begin{equation}
H_1 = \int \left(\mathcal{H}_{\text{ws}} + \mathcal{H}_J + u_1\,\phi_1\right) 
d^3x\,,
\label{eq:PrimaryHamiltonian}
\end{equation}
where $u_1$ is a Lagrange multiplier enforcing the primary constraint.

%==============================================================================

\subsection{Secondary constraint and the tertiary-stage candidate}
\label{sec:SecondaryConstraints}

With the primary constrained surface defined and the primary Hamiltonian
constructed, the Dirac--Bergmann algorithm requires us to demand the time
preservation of the primary constraint, $\dot{\phi}_{1} \approx 0$. To
evaluate this and the subsequent consistency conditions, we use the
fundamental equal-time Poisson brackets
\begin{equation}
\{A_\mu(x), \pi^\nu(y)\}
= \delta^\nu_\mu\,\delta^3(x-y),
\label{eq:PoissonBracket_Api}
\end{equation}
and
\begin{equation}
\{F_{ij}(x), \pi_k(y)\}
= \delta^k_i\,\partial_j\delta^3(x-y)
 - \delta^k_j\,\partial_i\delta^3(x-y).
\label{eq:PoissonBracket_Fpi}
\end{equation}

Using Eq.~\eqref{eq:PrimaryConstraint} and the explicit form of
Eq.~\eqref{eq:PrimaryHamiltonian}, the condition $\{\phi_1, H_1\}=0$ yields
the secondary constraint
\begin{equation}
\phi_2
= \partial_i\pi^i - J^0
+ \frac{1}{2}\,\theta^{kl}\,\partial_l(A_k J^0)
\approx 0.
\label{eq:SecondaryConstraint}
\end{equation}
In the commutative limit, this reduces to Gauss's law,
$\partial_i\pi^i=J^0$.

To test whether the algorithm closes or generates further constraints, we
must also preserve $\phi_2$ in time. Because the external current $J^\mu(x)$
is a prescribed function (not a dynamical variable), its time derivative
contributes through an explicit partial derivative. Hence
\begin{equation}
\dot\phi_2
= \{\phi_2,H_c\}_{\text{PB}}
+ \frac{\partial\phi_2}{\partial t}\bigg|_{\rm expl},
\label{eq:PhiDotSplit}
\end{equation}
with
\begin{equation}
\frac{\partial\phi_2}{\partial t}\bigg|_{\rm expl}
= -\partial_0J^0
+ \frac{1}{2}\theta^{kl}\partial_l\!\left(A_k\,\partial_0J^0\right).
\label{eq:PhiExpl}
\end{equation}

The canonical Hamiltonian splits as $H_c = H_{\text{ws}} + H_J$ (source-free and source parts), so the bracket decomposes accordingly:
\begin{equation}
\{\phi_2,H_c\}_{\text{PB}}
= \{\phi_2,H_{\text{ws}}\}_{\text{PB}}
+ \{\phi_2,H_J\}_{\text{PB}}.
\label{eq:PBsplit}
\end{equation}

Using the functional-derivative form of the Poisson bracket, with the intermediate contractions collected in Appendix~\ref{app:PBclosure}, one finds
\begin{equation}
\begin{aligned}
\{\phi_2,H_{\text{ws}}\}
{}&= \frac{1}{2}\theta^{kl}(\partial_lJ^0)\,\partial_kA_0 \\
&\quad
- \frac{1}{2}\theta^{jl}\eta_{00}\eta_{jn}
\,\partial_l\!\bigl(J^0\pi^n\bigr),
\end{aligned}
\label{eq:PBws_result}
\end{equation}
and
\begin{align}
\{\phi_2,H_J\}
{}&=-\partial_iJ^i
-\frac{1}{2}\theta^{il}\partial_i(J^0\partial_lA_0)
\nonumber\\
&\quad
+\frac{1}{2}\theta^{jl}\partial_i\partial_l(J^iA_j)
-\frac{1}{2}\theta^{il}\partial_i(J^k\partial_lA_k)
\nonumber\\
&\quad
-\frac{1}{2}\theta^{il}\eta_{00}\eta_{lk}\partial_i(J^0\pi^k)
+\frac{1}{2}\theta^{il}\partial_i(J^kF_{kl}).
\label{eq:PBJ_result}
\end{align}

Combining Eqs.~\eqref{eq:PBws_result} and \eqref{eq:PBJ_result} with
Eq.~\eqref{eq:PhiExpl} gives
\begin{align}
\dot\phi_2
={}& -\partial_\mu J^\mu
+ \theta^{li}(\partial_iJ^0)\partial_lA_0
+ \frac{1}{2}\theta^{jl}\partial_i\!\bigl(\partial_l(J^iA_j)\bigr)
\nonumber\\
&\quad
+ \theta^{li}\partial_i\!\bigl(J^k\partial_lA_k\bigr)
+ \frac{1}{2}\theta^{il}\partial_i\!\bigl(J^k\partial_kA_l\bigr)
\nonumber\\
&\quad+ \frac{1}{2}\theta^{kl}\partial_l\!\left(A_k\,\partial_0J^0\right).
\label{eq:PBintermediate}
\end{align}

If current conservation is not imposed a priori, the consistency condition
\begin{equation}
\Phi_3^{\rm cand}(x)
:=
\dot\phi_2(x)
\approx0
\label{eq:Phi3candDef}
\end{equation}
defines the third-stage candidate,
\begin{equation}
\begin{aligned}
\Phi_3^{\rm cand}
={}& -\partial_\mu J^\mu
+ \theta^{li}(\partial_iJ^0)\partial_lA_0 \\
&\quad
+ \theta^{li}\partial_i\!\bigl(J^k\partial_lA_k\bigr) \\
&\quad
+ \frac{1}{2}\theta^{jl}\partial_i\!\bigl(\partial_l(J^iA_j)\bigr) \\
&\quad
+ \frac{1}{2}\theta^{il}\partial_i\!\bigl(J^k\partial_kA_l\bigr) \\
&\quad
+ \frac{1}{2}\theta^{kl}\partial_l\!\left(A_k\,\partial_0J^0\right).
\end{aligned}
\label{eq:Phi3candExplicit}
\end{equation}
At $\theta=0$, this reduces to
\begin{equation}
\Phi_3^{\rm cand}\big|_{\theta=0}
= -\partial_\mu J^\mu,
\label{eq:Phi3candCommutative}
\end{equation}
so the tertiary-stage candidate collapses to the standard
current-compatibility condition. The first-order deformation of this
condition---and, as we now show, its Lagrangian meaning---is the central
technical result of the paper.

\subsection{\texorpdfstring{Main identity: tertiary candidate as the canonical pullback of $\partial_\nu \mathcal{E}_B^\nu$}{Main identity: tertiary candidate as the canonical pullback of the mapped divergence}}
\label{sec:MainIdentity}

The goal of this subsection is to establish the main identity of the paper:
the tertiary-stage candidate $\Phi_3^{\rm cand}=\dot\phi_2$ obtained above by
preservation of the secondary Gauss-type constraint coincides, once expressed
in the same $\mathcal{O}(\theta)$ canonical variables, with the divergence
$\partial_\nu\mathcal{E}_B^\nu$ of the mapped Euler--Lagrange equations. This
identity fixes the canonical locus at which the source-induced obstruction of
the action-level Banerjee realization enters the Dirac chain.

\begin{proposition}
\label{prop:DivELIdentity}
To first order in $\theta$, and without imposing
$\partial_\mu J^\mu=0$, the third-stage object generated by preserving the
secondary Gauss-type constraint satisfies the exact algebraic identity
\begin{equation}
\boxed{\;\partial_\nu \mathcal E_B^\nu
= \dot\phi_2
= \Phi_3^{\rm cand}\;}
\label{eq:MainIdentityDivEL}
\end{equation}
once the divergence of the mapped Euler--Lagrange equations is rewritten in
the same canonical phase-space variables used in the Dirac analysis and the
comparison is consistently truncated at first order in $\theta$. In the
action-level Banerjee realization considered here, the tertiary-stage
candidate is therefore algebraically identical to the canonical phase-space
pullback of the divergence of the mapped Euler--Lagrange equations. This
equality does not imply that the two objects have the same geometric status;
it states only that, when expressed in the same $\mathcal{O}(\theta)$
canonical framework, they encode the same source-compatibility content.
\end{proposition}

\noindent \emph{Proof sketch.} A direct Lagrangian computation is combined
with an independent canonical computation. On the Lagrangian side, starting
from the Banerjee-mapped Lagrangian in Eq.~\eqref{eq:LagrangianSplit},
\begin{equation}
\mathcal L_B
= \mathcal L_{\text{ws}}(F)+\mathcal L_J(A,\partial A;J),
\label{eq:LBsplit}
\end{equation}
the Euler--Lagrange operator decomposes as
\begin{equation}
\mathcal E_B^\nu
= \mathcal E_{\text{ws}}^\nu+\mathcal E_J^\nu,
\label{eq:ELsplit}
\end{equation}
with source contribution
\begin{equation}
\begin{aligned}
\mathcal E_J^\nu
={}& -J^\nu
-\theta^{\nu\alpha}J^\beta F_{\alpha\beta}
-\theta^{\alpha\beta}(\partial_\alpha J^\nu)A_\beta \\
&\quad
-\theta^{\alpha\beta}J^\nu\partial_\alpha A_\beta
+\frac{1}{2}\theta^{\nu\beta}(\partial_\mu J^\mu)A_\beta.
\end{aligned}
\label{eq:EJ_source_only}
\end{equation}
Since $\mathcal L_{\text{ws}}$ depends on $A_\mu$ only through the antisymmetric
field strength $F_{\mu\nu}$, its Euler--Lagrange operator has the form
\begin{equation}
\mathcal E_{\text{ws}}^\nu=\partial_\mu P^{\mu\nu},
\qquad
P^{\mu\nu}:=\frac{\partial \mathcal L_{\text{ws}}}{\partial(\partial_\mu A_\nu)},
\label{eq:ELwsdivform}
\end{equation}
with $P^{\mu\nu}=-P^{\nu\mu}$. Therefore
\begin{equation}
\partial_\nu \mathcal E_{\text{ws}}^\nu=0,
\label{eq:ELwsdivzero}
\end{equation}
and hence
\begin{equation}
\partial_\nu \mathcal E_B^\nu
= \partial_\nu \mathcal E_J^\nu.
\label{eq:divELsourceonly}
\end{equation}
On the canonical side, the computation of $\dot\phi_2$ carried out in
Sec.~\ref{sec:SecondaryConstraints} yields Eq.~\eqref{eq:Phi3candExplicit}.
Using the antisymmetry of $\theta^{ij}$ to rearrange terms, both
$\partial_\nu\mathcal{E}_B^\nu$ and $\dot\phi_2$ reduce to the common
expression
\begin{equation}
\begin{aligned}
\Phi_3^{\rm cand}
={}& -\partial_\mu J^\mu
+\theta^{li}(\partial_iJ^0)\partial_lA_0
+\theta^{li}\partial_i(J^k\partial_lA_k) \\
&\quad
-\frac{1}{2}\theta^{il}\partial_i\!\bigl[A_l(\partial_\mu J^\mu)\bigr].
\end{aligned}
\label{eq:CompatibilityCompact}
\end{equation}
A fully independent derivation---which does not rely on the compact form
above and uses only Eqs.~\eqref{eq:PhiDotSplit}--\eqref{eq:PhiExpl} together
with the functional derivatives of $\phi_2$---is provided in
Appendix~\ref{app:ELDivergenceIdentity}. \hfill$\square$

Equation~\eqref{eq:CompatibilityCompact} fixes the Lagrangian meaning of the
third-stage candidate: within the action-level Banerjee realization,
$\Phi_3^{\rm cand}$ is the canonical phase-space pullback of
$\partial_\nu\mathcal{E}_B^\nu$, without implying any stronger equivalence
between their geometric roles. An important structural remark is that this
identity is \emph{map-dependent}: the $F_{\alpha\beta}$-dependent terms
specific to the Banerjee parametrization $(c_1,c_2,c_3)=(1,\tfrac{1}{2},0)$
enter the explicit form of $\Phi_3^{\rm cand}$ through the source coupling
$\mathcal{L}_J$. For a different SW current map---such as the minimal map of
Adorno \emph{et al.}~\cite{Adorno2011CNE}---the same Dirac--Bergmann
procedure would produce a tertiary candidate with different
$\mathcal{O}(\theta)$ coefficients. The identity
\eqref{eq:MainIdentityDivEL} therefore characterizes the canonical
fingerprint of the Banerjee realization specifically, rather than a universal
feature of any action-level SW-mapped theory with external sources.

Equation~\eqref{eq:MainIdentityDivEL} also permits a precise comparison with
the canonical external-source analysis of Cabo and
Shabad~\cite{CaboShabad}: in their construction, the source-dependent
condition selected by consistency appears at the secondary stage, whereas in
the present framework the corresponding source-compatibility content is
generated at the tertiary stage through $\Phi_3^{\rm cand}$. The comparison
is therefore functional rather than stage-by-stage: the two objects are not
being identified with one another, but in both analyses constraint
preservation isolates the source-dependent condition required for consistency
in the presence of external currents.

%==============================================================================

\subsection{\texorpdfstring{Preservation of $\Phi_3^{\rm cand}$ and generic closure by multiplier fixing}{Preservation of the tertiary candidate and generic closure by multiplier fixing}}
\label{sec:Phi3CandEvolution}

Having identified $\Phi_3^{\rm cand}$ as the next consistency condition, the
Dirac algorithm requires
\begin{equation}
\begin{aligned}
\dot\Phi_3^{\rm cand}(x)
={}& \{\Phi_3^{\rm cand}(x),H_T\} \\
&\quad +\frac{\partial\Phi_3^{\rm cand}(x)}{\partial t}\Big|_{\rm expl}
\approx0.
\end{aligned}
\label{eq:Phi3CandEvolutionStart}
\end{equation}
Two structural facts are immediate. Since $\Phi_3^{\rm cand}$ contains no
canonical momenta,
\begin{equation}
\{\Phi_3^{\rm cand}(x),\Phi_3^{\rm cand}(y)\}=0,
\label{eq:Phi3CandSelfBracket}
\end{equation}
and the only $A_0$ dependence of $\Phi_3^{\rm cand}$ is through the term
$\theta^{li}(\partial_iJ^0)\partial_lA_0$, so its bracket with the primary
constraint is
\begin{equation}
\begin{aligned}
\{\Phi_3^{\rm cand}(x),\phi_1(y)\}
&= \Theta^l(x)\,\partial_l^x\delta^3(x-y), \\
\Theta^l &:= \theta^{li}\partial_iJ^0.
\end{aligned}
\label{eq:Phi3CandPrimaryBracket}
\end{equation}
Therefore the generic structure of the next consistency equation is
\begin{equation}
\begin{aligned}
\dot\Phi_3^{\rm cand}
={}&\Theta^l\partial_l u_1
+\mathcal F[A_\mu,\pi^\mu;J^\mu,\partial J,\partial^2J] \\
&\quad -\partial_0(\partial_\mu J^\mu)
\approx0,
\end{aligned}
\label{eq:Phi3CandEvolutionGeneric}
\end{equation}
where $\mathcal F$ collects the remaining first-order terms. The explicit
expression is recorded in Appendix~\ref{app:SelfBrackets}. Generically, for
$\Theta^l\neq0$, Eq.~\eqref{eq:Phi3CandEvolutionGeneric} fixes the primary
multiplier $u_1$ along the differential operator $\Theta^l\partial_l$ rather
than generating an independent quaternary constraint. In other words, the
generic inhomogeneous branch closes by multiplier fixing at the next stage.
Only in degenerate sectors, such as $\Theta^l=0$ together with additional
cancellations, can the algorithm continue further.

\subsection{Equal-time algebra and hierarchical source-kernel structure}
\label{sec:SectorStructure}

For the unrestricted external-current problem, it is useful to record the
equal-time algebra established up to the tertiary stage. Besides
$\{\phi_1,\phi_1\}=0$ and $\{\phi_1,\phi_2\}=0$, one has
\begin{equation}
\begin{aligned}
\{\phi_2(x),\phi_2(y)\}
={}&\, \frac{1}{2}\,\Theta^k(y)\,\partial_k^y\delta^3(x-y) \\
&\quad -\frac{1}{2}\,\Theta^k(x)\,\partial_k^x\delta^3(x-y),
\end{aligned}
\label{eq:phi2phi2_main}
\end{equation}
while Eq.~\eqref{eq:Phi3CandPrimaryBracket} gives, by antisymmetry,
\begin{equation}
\{\phi_1(x),\Phi_3^{\rm cand}(y)\}
=
-\Theta^l(y)\,\partial_l^y\delta^3(x-y),
\label{eq:phi1Phi3cand_main}
\end{equation}
and
\begin{equation}
\begin{aligned}
\{\phi_2(x),\Phi_3^{\rm cand}(y)\}
{}&=
\Sigma^{lk}(y)\,\partial_l^y\partial_k^y\delta^3(x-y) \\
&+\Xi^l(y)\,\partial_l^y\delta^3(x-y),
\label{eq:phi2Phi3cand_main}
\end{aligned}
\end{equation}
with
\begin{equation}
\begin{aligned}
\Sigma^{lk}
&:= \frac{1}{2}\Bigl(\theta^{li}\partial_i J^k+\theta^{ki}\partial_i J^l\Bigr), \\
\Xi^l
&:= -\frac{1}{2}\,\theta^{il}\,\partial_i(\partial_\mu J^\mu),
\end{aligned}
\label{eq:SigmaXiDefs_main}
\end{equation}
together with
\begin{equation}
\{\Phi_3^{\rm cand}(x),\Phi_3^{\rm cand}(y)\}=0.
\label{eq:Phi3candSelf_main}
\end{equation}

Equations~\eqref{eq:phi2phi2_main}--\eqref{eq:SigmaXiDefs_main} reveal that the
source dependence of the tertiary stage is not encoded in a single
obstruction, but in three distinct kernels. The kernel $\Theta^l$ controls
both the secondary self-bracket and the primary--tertiary bracket, whereas
the mixed secondary--tertiary bracket depends on $\Sigma^{lk}$ and $\Xi^l$.
These kernels probe different aspects of the prescribed source profile:
$\Theta^l$ measures the variation of the charge density along the
non-commutative directions, $\Sigma^{lk}$ measures the symmetrized spatial
gradients of the current flow, and $\Xi^l$ measures the spatial gradient of
the source divergence. They are not redundant: ordinary current conservation
forces $\Xi^l=0$ but does not by itself imply either $\Theta^l=0$ or
$\Sigma^{lk}=0$; likewise, spatial constancy of $J^k$ may set $\Sigma^{lk}=0$
without constraining the remaining kernels.

This observation induces a hierarchical organization of source sectors. The
generic non-degenerate unrestricted branch is characterized by
$\Theta^l\neq0$. In that case, the source-induced obstruction is already
visible in $\{\phi_2,\phi_2\}$ and $\{\phi_1,\Phi_3^{\rm cand}\}$, and
preservation of $\Phi_3^{\rm cand}$ constrains the primary multiplier
through Eq.~\eqref{eq:Phi3CandEvolutionGeneric}; subject to the properties of
the differential operator $\Theta^l\partial_l$, the algorithm closes by
multiplier determination rather than by an independent quaternary constraint.
No unrestricted first-class interpretation is supported in this sector. By
contrast, when $\Theta^l=0$ but at least one of $\Sigma^{lk}$ or $\Xi^l$
remains non-vanishing, the obstruction disappears from
$\{\phi_2,\phi_2\}$ and $\{\phi_1,\Phi_3^{\rm cand}\}$ but survives in the
mixed bracket $\{\phi_2,\Phi_3^{\rm cand}\}$, so the vanishing of $\Theta^l$
alone is not sufficient to recover a first-class structure.

At the stage reached here, the manuscript does \emph{not} claim a complete
first-/second-class classification of the unrestricted source sector. What
has been established is the equal-time algebra available up to the tertiary
stage and the associated consistency structure, which should be read as the
canonical information presently available for the unrestricted problem rather
than as an exhaustive classification of the external-current theory.

\noindent\textbf{Regime change and outlook for the reduced description.}
The analysis of Secs.~\ref{sec:Phi3CandEvolution}--\ref{sec:SectorStructure}
produces two structurally distinct outcomes. For the generic non-degenerate
inhomogeneous branch $\Theta^l\neq 0$, preservation of $\Phi_3^{\rm cand}$
determines the primary multiplier along $\Theta^l\partial_l$ rather than
yielding an independent quaternary constraint, and no standard first-class
classification is available without further input. For source profiles
satisfying simultaneously
\begin{equation}
\Theta^l=0,\qquad
\Sigma^{lk}=0,\qquad
\Xi^l=0,\qquad
\Phi_3^{\rm cand}\equiv0,
\label{eq:FirstClassSectorConditions}
\end{equation}
the displayed equal-time obstructions disappear, the surviving constraint
pair $(\phi_1,\phi_2)$ becomes first-class in the strict equal-time sense,
and the standard machinery of gauge generators, Dirac brackets, and reduced
dynamics applies unchanged.

To see what these conditions select, consider the concrete case of purely
$x^1$--$x^2$ non-commutativity, $\theta^{12}=-\theta^{21}=\theta$ with all
other components zero. The kernel $\Theta^l=\theta^{li}\partial_iJ^0$ then
requires $\partial_1 J^0=\partial_2 J^0=0$: the charge density must be
uniform across the non-commutative plane. Likewise, $\Sigma^{lk}=0$ forces
all spatial current components to be independent of $x^1$ and $x^2$, and
$\Xi^l=0$ demands the same of $\partial_\mu J^\mu$. Any source of the form
\begin{equation}
J^\mu=\bigl(\rho(t,x^3),\,0,\,0,\,j(t,x^3)\bigr),
\qquad
\partial_0\rho+\partial_3 j=0,
\label{eq:AdmissibleSourceExample}
\end{equation}
automatically satisfies the first three conditions in
Eq.~\eqref{eq:FirstClassSectorConditions}. Under these conditions, the
remaining requirement $\Phi_3^{\rm cand}\equiv 0$ reduces to ordinary current
conservation $\partial_\mu J^\mu=0$, which is precisely the continuity
equation already imposed in Eq.~\eqref{eq:AdmissibleSourceExample}. The
physical picture is therefore a conserved charge--current system whose
profile may vary along the direction parallel to
$\boldsymbol{\theta}$ but is uniform across the non-commutative plane. In
particular, the sector conditions do not require $J^\mu$ to vanish: they
restrict its \emph{derivative profile} (spatial gradients and divergence) in
the non-commutative directions, while the pointwise values $J^0$, $J^k$
remain non-trivial. This distinction is important for reading the reduced
Hamilton equations of Sec.~\ref{sec:ReducedDynamics}, where $J^\mu$ appears
in its general functional form even though the brackets through which the
equations are derived exist only for source profiles of the type
\eqref{eq:AdmissibleSourceExample}.

The remainder of the paper,
Sec.~\ref{sec:ReducedDescription}, develops this reduced description in that
restricted subcase and uses it to confront the Banerjee route with the
gauge-invariant equation-level formulation of Ref.~\cite{Adorno2011CNE}. No
claim about the unrestricted external-current theory is made in what follows;
the sectorial caveat is declared once here and not repeated subsection by
subsection.

\section{Reduced phase-space description in the restricted first-class sector}
\label{sec:ReducedDescription}
%==============================================================================

All constructions in this section presuppose the sector conditions of
Eq.~\eqref{eq:FirstClassSectorConditions}, under which the surviving
equal-time constraint pair is $(\phi_1,\phi_2)$ and no further independent
consistency condition arises at the level analyzed in
Sec.~\ref{sec:DiracBergmann}. The reduced Hamiltonian structure, gauge
generator, Dirac brackets and reduced dynamics obtained below are therefore
claimed only for this restricted sufficient subcase and are \emph{not}
presented as statements about the unrestricted external-current problem.

\subsection{Sector conditions, gauge generator, and degree-of-freedom count}
\label{sec:GaugeGeneratorDOF}

Within the restricted sector defined by
Eq.~\eqref{eq:FirstClassSectorConditions}, the gauge generator associated with
the surviving first-class pair $(\phi_1,\phi_2)$ can be introduced in the
standard Castellani-type form
\begin{equation}
G[\epsilon_0,\epsilon]
=
\int d^3x\,\bigl[\epsilon_0(x)\phi_1(x)+\epsilon(x)\phi_2(x)\bigr],
\label{eq:GaugeGenerator}
\end{equation}
where $\epsilon_0$ and $\epsilon$ are arbitrary test functions at fixed time.
Using the canonical equal-time brackets, the gauge variations of the basic
variables read
\begin{align}
\delta A_0 &= \epsilon_0,
&
\delta A_i &= \partial_i\epsilon,
\nonumber\\
\delta \pi^0 &= 0,
&
\delta \pi^i &= \frac{1}{2}\theta^{ik}\partial_k(\epsilon J^0),
\label{eq:GaugeDeltaPi}
\end{align}
and reduce to the standard canonical $U(1)$ transformations in the commutative
limit. These transformations describe the residual equal-time gauge freedom of
the dynamical sector in the restricted setting.

The independent physical variables are the transverse components of the gauge
field and their conjugate momenta. Starting from the eight canonical field
variables $(A_\mu,\pi^\mu)$, the surviving pair $(\phi_1,\phi_2)$ provides
two first-class constraints, so that the reduced number of physical
phase-space degrees of freedom is
\begin{equation}
N_{\mathrm{phys}}
=\frac{1}{2}(8-2\times 2)=2,
\label{eq:DOFcount}
\end{equation}
corresponding to two propagating configuration-space degrees of freedom. The
explicit gauge fixing, the corresponding Dirac reduction, and the resulting
reduced brackets are addressed in Sec.~\ref{sec:ReducedBrackets}.

\subsection{Dirac brackets in Coulomb gauge}
\label{sec:ReducedBrackets}

To perform the reduction, we impose the Coulomb-type gauge-fixing conditions
\begin{equation}
\chi_1 := A_0 \approx 0,
\qquad
\chi_2 := \partial_i A^i \approx 0.
\label{eq:GaugeFixing}
\end{equation}
Together with the surviving pair $(\phi_1,\phi_2)$, these conditions define
the second-class set
\begin{equation}
\Gamma_a=(\phi_1,\phi_2,\chi_1,\chi_2),
\label{eq:SecondClassSet}
\end{equation}
with
\begin{equation}
\phi_1=\pi^0,\qquad
\phi_2=\partial_i\pi^i-J^0+\frac{1}{2}\theta^{kl}\partial_l(A_kJ^0).
\label{eq:FirstClassPairRepeated}
\end{equation}
In the restricted sector, the non-vanishing equal-time brackets among these
quantities are the standard ones,
\begin{equation}
\begin{aligned}
\{\phi_1(x),\chi_1(y)\}&=-\delta^3(x-y), \\
\{\phi_2(x),\chi_2(y)\}&=-\nabla^2\delta^3(x-y),
\label{eq:ConstraintGaugeBrackets}
\end{aligned}
\end{equation}
while the source-dependent obstructions displayed in the unrestricted problem
have already vanished. The corresponding constraint matrix is therefore
invertible in the standard equal-time sense.

The associated Dirac bracket is
\begin{equation}
\begin{aligned}
\{F,G\}_D
={}&
\{F,G\}
-\int d^3u\,d^3v\,
\{F,\Gamma_a(u)\}\nonumber\\
&\qquad\qquad\,C^{-1}_{ab}(u,v)\,\{\Gamma_b(v),G\},
\label{eq:DiracBracketDefinition}
\end{aligned}
\end{equation}
where $C_{ab}(x,y)=\{\Gamma_a(x),\Gamma_b(y)\}$ and $C^{-1}_{ab}$ denotes
its inverse. The resulting brackets among the independent transverse variables
take the familiar Coulomb-gauge form,
\begin{equation}
\{A_i(x),A_j(y)\}_D=0,
\label{eq:DBAA}
\end{equation}
\begin{equation}
\{\pi_i(x),\pi_j(y)\}_D=0,
\label{eq:DBPiPi}
\end{equation}
and
\begin{equation}
\{A_i(x),\pi_j(y)\}_D
=
\left(\delta_{ij}
-\frac{\partial_i\partial_j}{\nabla^2}\right)\delta^3(x-y).
\label{eq:DBAPi}
\end{equation}
Accordingly, only the transverse components of the gauge field and canonical
momentum remain as independent reduced variables. These brackets provide the
canonical input for the reduced Hamiltonian evolution discussed in the next
subsection.

\subsection{Reduced Hamilton equations and Maxwell-like rewriting}
\label{sec:ReducedDynamics}

With the Dirac brackets in place, the reduced Hamiltonian generates the
equal-time evolution of the independent phase-space variables through
\begin{align}
\dot{A}_s &= \{A_s, H_c\}_*\,, \label{eq:AdotGeneral} \\
\dot{\pi}^s &= \{\pi^s, H_c\}_*\,. \label{eq:PidotGeneral}
\end{align}
A remark on the logical status of the source terms is in order. The Dirac
bracket $\{,\}_*$ constructed in Sec.~\ref{sec:ReducedBrackets} exists only
in the restricted sector~\eqref{eq:FirstClassSectorConditions}, because
the invertibility of the constraint matrix in its standard Coulomb-gauge form
requires $\Theta^l=\Sigma^{lk}=\Xi^l=0$. The canonical Hamiltonian
$H_c=H_{\text{ws}}+H_J$, however, retains the full source coupling: the
sector conditions constrain the admissible \emph{profiles} of $J^\mu$---its
spatial derivatives and divergence---but do not set $J^\mu$ itself to zero.
The resulting Hamilton equations therefore contain the external current in its
general functional form, while remaining valid only for source profiles
satisfying Eq.~\eqref{eq:FirstClassSectorConditions}. For concreteness, the
admissible class includes the conserved, NC-plane-uniform profiles
illustrated in Eq.~\eqref{eq:AdmissibleSourceExample}: the source terms
$J^0$, $J^k$ in the equations below are non-vanishing but have restricted
spatial gradients. The equations are presented in this unsimplified form to
make the comparison with the Euler--Lagrange system of
Sec.~\ref{sec:ELBanerjee} as direct as possible; no further reduction using
the sector conditions on the source terms is attempted here.

At first order in $\theta$, the reduced evolution is most compactly written as
\begin{equation}
\dot{A}_s = \eta^{sc}\,\Lambda_c\,,
\label{eq:Adot}
\end{equation}
with
\begin{align}
\Lambda^c
={}&
-\eta_{00}\Bigl[
\bigl(1 - \mathbf{B}\cdot\boldsymbol{\theta}\bigr)\pi^c
+ (\boldsymbol{\pi}\cdot\boldsymbol{\theta})B^c
+ (\mathbf{B}\cdot\boldsymbol{\pi})\theta^c\Bigr]
\nonumber\\
&+ \tfrac{1}{2}\eta_{00}J^0(\mathbf{A}\times\boldsymbol{\theta})^c\,.
\label{eq:LambdaVector}
\end{align}
For the momentum one finds
\begin{align}
\dot{\pi}^s
&= -\varepsilon^{nsa}\partial_n\Bigl[
\bigl(1 + \mathbf{B}\!\cdot\!\boldsymbol{\theta}\bigr)F^{ns}
+ (\boldsymbol{\pi}\!\cdot\!\boldsymbol{\theta})\pi_a
\nonumber\\
&\qquad\quad\;
- \tfrac{1}{2}\bigl(\eta_{00}\pi_j\pi^j - B_jB^j\bigr)\theta_a\Bigr]
\nonumber\\
&\quad + J^0\theta^{sl}\pi_l
- J^j\theta^{sl}\partial_j A_l - J^s\,.
\label{eq:PidotVector}
\end{align}
Equations~\eqref{eq:Adot}--\eqref{eq:PidotVector} already show that the
reduced electric sector is not obtained from the Maxwell relation
$E^i=\pi^i$ alone, but from a source- and field-dependent deformation of that
identification at first order in $\theta$. In the commutative limit,
Eqs.~\eqref{eq:LambdaVector} and~\eqref{eq:PidotVector} reduce to the
standard Hamilton equations of Maxwell theory with external sources.

Two structural features of the source sector should be noted at the outset.
First, the source terms in Eqs.~\eqref{eq:LambdaVector}
and~\eqref{eq:PidotVector} contain the gauge potential $A_j$ explicitly---not
only through the field strength $F_{jl}$. In $\Lambda^c$, the term
$\tfrac{1}{2}\eta_{00}J^0(\mathbf{A}\times\boldsymbol{\theta})^c$ depends on
$A_j$ itself; in $\dot\pi^s$, the combination
$-J^j\theta^{sl}\partial_jA_l$ involves $\partial_jA_l$ rather than
$F_{jl}=\partial_jA_l-\partial_lA_j$. These terms are the Hamiltonian
counterpart of the gauge-non-covariant source coupling identified at the
Lagrangian level by Adorno \emph{et al.}~\cite{Adorno2011CNE}. Within the
present gauge-fixed formulation ($A_0=0$, $\partial_kA^k=0$), the $A_j$ that
appears is the physical transverse potential and the equations are
self-consistent; but the appearance of the undifferentiated potential---rather
than the field strength---in the source terms signals that the underlying
action-level realization is not gauge covariant when external currents are
present. The precise decomposition of this non-covariance into a
scheme-dependent gauge-invariant piece and an explicitly gauge-dependent piece
is carried out in the on-shell comparison of Sec.~\ref{sec:OnShellComparison}.
Second, the field-strength bilinear sector (all terms not involving $J^\mu$) is
manifestly gauge invariant, as it depends only on $F_{ij}$, $B^i$, and
$\pi^i$. The asymmetry between gauge-invariant kinematics and
gauge-non-covariant source coupling will be reflected in the on-shell
comparison, where the Gauss sector is recovered exactly while the Amp\`ere
sector retains irreducible residues.

\noindent\textbf{Maxwell-like rewriting of the reduced dynamics.}
The reduced Hamilton equations admit a Maxwell-like rewriting that is
convenient for the comparison with the mapped Euler--Lagrange system carried
out in Secs.~\ref{sec:ELBanerjee}--\ref{sec:OnShellComparison}. Using
Eq.~\eqref{eq:EB} and the solution~\eqref{eq:F0i_solution} for $F^{0i}$, the
electric field may be written as
\begin{equation}
\begin{aligned}
E^i
={}&\, \pi^k F_{kl}\theta^{li}
+ \pi^k F_{ml}\eta^{im}\eta_{nk}\theta^{ln}
\\
&\quad - \pi^i\left(\tfrac{1}{2}F_{kl}\theta^{lk}-1\right)
+ \tfrac{1}{2}\,\theta^{ji}A_j J^0\,.
\end{aligned}
\label{eq:ElectricField}
\end{equation}
Note that the last term, $\tfrac{1}{2}\theta^{ji}A_jJ^0$, makes the
effective electric field explicitly dependent on the gauge potential in the
presence of sources; in the source-free limit $J^0\to0$ this term vanishes
and $E^i$ becomes a function of $F_{\mu\nu}$ and $\pi^k$ alone.
Hence
\begin{equation}
\dot{A}_s = \eta_{00}E_s\,.
\label{eq:AdotEqualsE}
\end{equation}

The magnetic equations remain kinematic:
\begin{equation}
\partial_0\mathbf{B} = -\nabla\times\mathbf{E}\,,
\label{eq:FaradayLaw}
\end{equation}
and
\begin{equation}
\nabla\cdot\mathbf{B} = 0\,.
\label{eq:GaussLawMagnetism}
\end{equation}

The physical Gauss law should not be inferred from the divergence of
Eq.~\eqref{eq:AdotEqualsE}. That equation governs the time evolution of the
spatial gauge potential $A_i$, whereas Gauss's law belongs to the canonical
constraint structure of the theory. In the present Hamiltonian formulation,
the Gauss condition is encoded directly in the secondary constraint obtained
from preserving the primary constraint $\pi^0\approx0$, namely
\begin{equation}
\partial_i\pi^i
=
J^0 - \frac{1}{2}\,\theta^{kl}\,\partial_l(A_k J^0)\,.
\label{eq:GaussConstraint}
\end{equation}
Thus, Gauss's law is not an independent consequence of the evolution equation
for $A_i$, but a restriction that defines the admissible phase-space
configurations at each time slice. Its role is to constrain the longitudinal
part of the canonical momentum in the presence of the prescribed external
charge density, including the first-order non-commutative correction.

Identifying the reduced electric-displacement variable with the canonical
momentum,
\begin{equation}
D^s \equiv \pi^s,
\label{eq:DField}
\end{equation}
Eq.~\eqref{eq:PidotVector} can be rewritten in Maxwell-like form as
\begin{equation}
(\nabla\times\mathbf{H})^s = T^s - \partial_0 D^s,
\label{eq:AmpereMaxwellNC}
\end{equation}
with
\begin{align}
\mathbf{H}
&= \bigl[1 + B_j\theta^j\bigr]\mathbf{B}
+ \eta_{00}(\boldsymbol{\pi}\cdot\boldsymbol{\theta})\boldsymbol{\pi}
\nonumber\\
&\quad + \tfrac{1}{2}(B_a B^a)\boldsymbol{\theta}
- \tfrac{1}{2}\eta_{00}(E_j E^j)\boldsymbol{\theta}\,,
\label{eq:HField} \\
T^s
&= J^0\,\theta^{sl}\,\pi_l - J^j\,\theta^{sl}\,\partial_j A_l - J^s\,.
\label{eq:SourceTerm}
\end{align}
Equation~\eqref{eq:AmpereMaxwellNC} is not an independent constitutive law
but a compact rewriting of the reduced Hamilton equation for $\pi^s$; the
notation is organizational. The source sector $T^s$ contains the external
current in its general functional form; the sector conditions
\eqref{eq:FirstClassSectorConditions} constrain only the derivative profile
of $J^\mu$, not its pointwise values, so $J^0$, $J^j$, $J^s$ are not
simplified here. Even so, $T^s$ already contains both gauge-invariant and
explicitly gauge-dependent contributions, which anticipates the asymmetric
on-shell comparison with the Euler--Lagrange equations: the Gauss sector is
anchored directly in the secondary constraint, whereas the Amp\`ere-like
sector inherits the full source-dependent deformation of the reduced
Hamiltonian evolution.
The term $\eta_{00}(\boldsymbol{\pi}\cdot\boldsymbol{\theta})\boldsymbol{\pi}$
in $\mathbf{H}$ encodes the anisotropic deformation of the effective magnetic
permeability induced by the NC structure: with $\eta_{00}=+1$ and
$\boldsymbol{\pi}\to\mathbf{E}$ on shell, this becomes
$+(\mathbf{E}\cdot\boldsymbol{\theta})\mathbf{E}$, coupling the electric
sector to the NC direction $\boldsymbol{\theta}$ already at the constitutive
level. The commutative limit $\theta\to 0$ recovers $\mathbf{H}=\mathbf{B}$,
as required.

This Maxwell-like representation should not be overinterpreted. In
particular, the source sector $T^s$ is not gauge covariant: it contains
the gauge potential $A_j$ through the term $-J^j\theta^{sl}\partial_jA_l$,
which does not combine into $F_{jl}$. Accordingly, the Amp\`ere-like
equation~\eqref{eq:AmpereMaxwellNC} does not have the same gauge-invariance
status as its commutative counterpart. The representation is nevertheless
useful because it separates the gauge-invariant field-strength block
($\mathbf{H}$, $D^s$) from the gauge-non-covariant source block ($T^s$) in a
form that makes the comparison with the mapped Euler--Lagrange structure
(Sec.~\ref{sec:ELBanerjee}) and the on-shell gauge-covariance analysis
(Sec.~\ref{sec:OnShellComparison}) as direct as possible. Its meaning is
therefore internal to the restricted sufficient first-class subcase.

\subsection{Comparison with the Banerjee Euler--Lagrange system (off-shell)}
\label{sec:ELBanerjee}

Having rewritten the reduced Hamilton equations in a Maxwell-like form, we
now compare them with the Euler--Lagrange equations of the mapped theory.
The point of the comparison is limited but structural: it tests how far the
reduced canonical dynamics reproduces the field-equation content of the
mapped Banerjee theory once the source-dependent obstructions have been
removed by the sector conditions of Eq.~\eqref{eq:FirstClassSectorConditions}.
The logical relation among the relevant formulations is
\begin{equation}
\begin{gathered}
\underbrace{\dot\pi^s = \{\pi^s,H_c\}_*}_{\text{Hamilton (this work)}}
\;\stackrel{\text{Legendre}}{\longleftrightarrow}\;
\underbrace{\mathcal{E}^\nu_B = 0}_{\text{Banerjee EL}}
\\[4pt]
\longrightarrow\;
\underbrace{\mathcal{E}^\nu_A = 0}_{\text{Adorno Eq.~(24)}}
\;\dashrightarrow\;
\underbrace{\text{Adorno Eqs.~(16)--(17)}}_{\text{gauge-inv.\ target}}
\end{gathered}
\label{eq:LogicalChain}
\end{equation}
Here the solid arrows denote exact equivalence or direct comparison inside
the restricted branch, whereas the dashed arrow denotes only a partial
on-shell correspondence between the action-level Banerjee route and the
gauge-covariant equation-level route of Adorno \emph{et al.}

Applying the Euler--Lagrange operator to
$\mathcal{L}=\mathcal{L}_{\text{ws}}+\mathcal{L}_J$, one obtains
\begin{align}
\mathcal{E}^\nu_B
&= \partial_\mu F^{\mu\nu}
+ \theta^{\alpha\beta}\Bigl[
\tfrac{1}{4}\partial^\nu(F_{\alpha\mu}F_\beta{}^\mu)
- \partial_\mu(F^\mu_{\ \alpha}F^\nu_{\ \beta})
\nonumber\\ & + F_{\mu\alpha}\partial_\beta F^{\mu\nu}\Bigr]
- J^\nu
- \theta^{\nu\alpha}J^\beta F_{\alpha\beta}
- \theta^{\alpha\beta}(\partial_\alpha J^\nu)A_\beta
\nonumber\\
& - \theta^{\alpha\beta}J^\nu\partial_\alpha A_\beta
+ \tfrac{1}{2}\theta^{\nu\beta}(\partial_\mu J^\mu)A_\beta
= 0\,.
\label{eq:ELBanerjee}
\end{align}
The first two lines constitute the common field-strength bilinear sector
(same in both the Banerjee and Adorno formulations); the remaining
source-dependent terms contain explicit gauge potentials and are therefore
gauge-non-covariant for external fixed currents. The term
$\tfrac{1}{2}\theta^{\nu\beta}(\partial_\mu J^\mu)A_\beta$ vanishes
when the external current is conserved but must be retained in the general
off-shell expression.

For comparison, Adorno \emph{et al.} obtain, with the minimal current map,
\begin{align}
\mathcal{E}^\nu_A
&= \partial_\mu\Bigl[
  f^{\mu\nu}\Bigl(1 + \tfrac{g}{2}\theta^{\alpha\beta}f_{\alpha\beta}\Bigr)
\Bigr]
\nonumber\\
&\quad - g\theta^{\alpha\beta}\partial_\mu(f^{\mu\nu}f_{\alpha\beta})
- g\theta^{\alpha\beta}\partial_\mu(f^{\mu\alpha}f^{\nu\beta})
\nonumber\\
&\quad + g\theta^{\alpha\beta}\partial_\beta(f^{\mu\nu}f_{\mu\alpha})
- \tfrac{g}{4}\theta^{\alpha\beta}\partial_\alpha(f_{\mu\nu}f^{\mu\nu})
\nonumber\\
&\quad - \frac{4\pi}{c}\,j^\nu\Bigl(1+\tfrac{g}{2}\theta^{\alpha\beta}f_{\alpha\beta}\Bigr)
\nonumber\\
&\quad - \frac{4\pi g}{c}\theta^{\alpha\beta}\Bigl[
  f_{\alpha\beta}j^\nu + f^\nu{}_{\alpha}j_\beta
\nonumber\\
&\qquad\quad
  + A_\alpha\partial_\beta j^\nu
  - \tfrac12 A^\nu\partial_\alpha j_\beta
\Bigr]
= 0\,.
\label{eq:ELAdorno24}
\end{align}
The field-strength sector agrees in both formulations because it comes from the
same SW-mapped Maxwell block $\mathcal{L}_{\text{ws}}$. The difference lies entirely
in the source sector: the Banerjee map generates, in addition to $-J^\nu$, the
gauge-non-covariant terms
$-\theta^{\nu\alpha}J^\beta F_{\alpha\beta}$,
$-\theta^{\alpha\beta}(\partial_\alpha J^\nu)A_\beta$, and
$-\theta^{\alpha\beta}J^\nu\partial_\alpha A_\beta$.
The last two contain the gauge potential explicitly and produce the
gauge-covariance obstruction identified by Adorno \emph{et al.}
Thus, the comparison with Adorno must be understood as a comparison between
different admissible SW-map choices, not as an identity between the two
Lagrangian formulations. In the present restricted sector, the reduced
Hamilton equations reproduce the common field-strength content of the mapped
Euler--Lagrange system, while the residual source-dependent mismatch measures
the limitation of the reduced Banerjee description rather than an inconsistency
of the unrestricted theory as a whole.

\subsection{On-shell comparison in the reduced regime}
\label{sec:OnShellComparison}

Section~\ref{sec:ELBanerjee} compared the reduced Hamilton equations with the
mapped Euler--Lagrange equations at the level of their explicit component
form. We now evaluate that comparison on shell, under the gauge conditions
$A_0=0$ and $\partial_kA^k=0$, working throughout to first order in $\theta$.
The temporal and spatial sectors do not behave symmetrically: the temporal
comparison is anchored directly in the secondary constraint, whereas the
spatial comparison inherits the full source-dependent deformation of the
reduced Hamiltonian evolution.

\noindent\textbf{Spatial component: partial correspondence.}
From Eq.~\eqref{eq:Pi_i_final}, the canonical momentum on shell is
\begin{equation}
\pi^i
=
F^{0i}\bigl[1+\mathcal{O}(\theta)\bigr]
+\tfrac12\,\theta^{ji}A_jJ^0\,.
\label{eq:piOnShell}
\end{equation}
Substituting into Eq.~\eqref{eq:PidotVector}, the field-strength sector yields
exactly the bilinear structure appearing on the left-hand side of the
gauge-invariant field equation of Ref.~\cite{Adorno2011CNE}. This part of the
correspondence is exact within the restricted sufficient first-class regime.

The source contribution is
\begin{equation}
T^s
=\, J^0\theta^{sl}\pi_l - J^j\theta^{sl}\partial_j A_l - J^s\,,
\label{eq:TsSource}
\end{equation}
which can be rewritten as
\begin{equation}
T^s
=\, \bigl(J^\mu\theta^{sl}F_{\mu l} - J^s\bigr) - J^j\theta^{sl}\partial_l A_j\,.
\label{eq:TsDecomposed}
\end{equation}
Two independent obstructions remain:

\emph{(i) Scheme-dependent gauge-invariant obstruction.}
The term $J^\mu\theta^{sl}F_{\mu l}$ is gauge invariant but does not appear
on the right-hand side of the gauge-invariant equation-level formulation of
Ref.~\cite{Adorno2011CNE}. It reflects the Banerjee-map choice and cannot be
removed by the constraint projection.

\emph{(ii) Gauge-dependent obstruction.}
The term $-J^j\theta^{sl}\partial_l A_j$ contains $A_j$ explicitly and
remains gauge dependent. Under $A_j\to A_j+\partial_j\Lambda$ it transforms as
\begin{equation}
\begin{aligned}
-J^j\theta^{sl}\partial_l A_j
&\;\longrightarrow\;
-J^j\theta^{sl}\partial_l A_j - J^j\theta^{sl}\partial_l\partial_j\Lambda\,,
\end{aligned}
\label{eq:GaugeDep}
\end{equation}
and the residual term does not vanish generically. Hence the spatial
Hamiltonian equation does not reduce to the gauge-invariant Amp\`ere equation
of Ref.~\cite{Adorno2011CNE}: what survives is an exact identification of
the common field-strength bilinear sector, but not a full reduction of the
spatial Hamiltonian equation to the gauge-invariant Amp\`ere equation.

\noindent\textbf{Temporal component: exact Gauss-law recovery.}
The temporal component is different. On the constraint surface,
\begin{equation}
\partial_i\pi^i
=
J^0 - \tfrac{1}{2}\theta^{kl}\partial_l(A_k J^0)\,.
\label{eq:GaussOnShell}
\end{equation}
Using the decomposition~\eqref{eq:KofFImplicit}--\eqref{eq:KofFDef} of the
spatial canonical momentum introduced in Sec.~\ref{sec:CurrentConsistency},
one obtains
\begin{equation}
\begin{split}
\partial_i\pi^i
={}& \partial_i F^{0i}
  + \partial_i \mathcal{K}^i(F)
  + \tfrac{1}{2}\theta^{ji}(\partial_i A_j)\,J^0
 \\ & + \tfrac{1}{2}\theta^{ji}A_j\,\partial_i J^0\,.
\end{split}
\label{eq:divpi_expanded}
\end{equation}
The two source-dependent terms in Eq.~\eqref{eq:divpi_expanded} cancel
exactly against the corresponding terms on the right-hand side of
Eq.~\eqref{eq:GaussOnShell}: the $(\partial_i A_j)J^0$ pair cancels by
antisymmetry of $\theta^{ji}$, and the $A_k\partial_l J^0$ pair cancels
identically. The remainder contains only the field-strength bilinear block
$\partial_i F^{0i}+\partial_i\mathcal{K}^i(F)$ and no explicit potential
terms, and therefore coincides with the $\mu=0$ component of the
gauge-invariant formulation of Ref.~\cite{Adorno2011CNE}. This is the exact
on-shell result within the restricted sufficient first-class regime.

\paragraph{NC-deformed conservation on the constraint surface.}
By Proposition~\ref{prop:DivELIdentity}, requiring the third-stage
consistency condition $\Phi_3^{\rm cand}\approx0$ is equivalent to requiring
that the mapped field equations be divergence-consistent once their
divergence is rewritten in the same canonical phase-space variables and
truncated consistently at first order in $\theta$. At first order this
amounts to a source-compatibility requirement generated on shell by the
algorithm, not to ordinary current conservation imposed from outside. In
the restricted sufficient first-class subcase studied here, no separate
reduced tertiary representative is introduced: one must instead impose
directly that the full tertiary condition vanish identically together with
the kernel conditions displayed in Sec.~\ref{sec:SectorStructure}.

Taken together, the main-text identity of Sec.~\ref{sec:MainIdentity} and
the on-shell comparison above locate the source-induced obstruction of the
action-level Banerjee realization in two complementary places: at the
\emph{tertiary} stage of the Dirac chain for the unrestricted problem, and
in the \emph{spatial Amp\`ere-type sector} of the reduced dynamics for the
restricted first-class subcase. The temporal sector, anchored directly in
the secondary constraint, reproduces the gauge-invariant Gauss law exactly
on shell; the spatial sector retains two irreducible residues---one
scheme-dependent and gauge invariant, inherited from the current-map
ambiguity, and one explicitly gauge dependent. The following section
examines what this canonical diagnosis implies for the Lagrangian mismatch
identified by Adorno \emph{et al.}, and delimits the scope of the present
results with respect to alternative SW current maps and action-level
realizations.

\section{Discussion}
\label{sec:Discussion}
%==============================================================================

\subsection{Canonical \emph{versus} Lagrangian locus of the obstruction}
\label{sec:Discussion-CanonicalLocus}

The analysis of Sec.~\ref{sec:DiracBergmann} assigns a precise canonical
meaning to the source-induced obstruction. In the unrestricted problem, the
preservation of the secondary Gauss-type constraint produces the third-stage
identity~\eqref{eq:MainIdentityDivEL}. Its zeroth-order part is the standard
compatibility term $\partial_\mu J^\mu$, while its full first-order content
is the NC-deformed source-compatibility condition generated on shell by the
algorithm. The crucial point is that this quantity is simultaneously the
next object in the Dirac chain and, after canonical pullback, the divergence
of the Banerjee Euler--Lagrange equations. The generic source sector
therefore closes by multiplier fixing rather than by the appearance of a
universal quaternary constraint, while first-class reduced behavior survives
only as a derived special subcase. The Hamiltonian analysis thus sharpens
the Lagrangian obstruction identified by Adorno \emph{et al.}: what appears
there as a failure of gauge covariance in the action-level route reappears
here as the explicit point at which source compatibility enters the Dirac
consistency chain. More precisely, the canonical formulation provides three
layers of information that are not accessible from the Lagrangian side alone:
it \emph{localizes} the obstruction at a definite stage of the Dirac chain
(the tertiary level, via the identity~\eqref{eq:MainIdentityDivEL}); it
\emph{classifies} the source dependence through the hierarchical kernel
structure $(\Theta^l,\Sigma^{lk},\Xi^l)$ of
Sec.~\ref{sec:SectorStructure}; and it \emph{decomposes} the on-shell
comparison into a Gauss sector that is recovered exactly and an Amp\`ere
sector that retains irreducible residues
(Sec.~\ref{sec:OnShellComparison}). None of these distinctions has a
direct counterpart in the equation-level analysis of
Ref.~\cite{Adorno2011CNE}.

The asymmetry between the temporal and spatial on-shell comparisons in
Sec.~\ref{sec:OnShellComparison} is not incidental. The temporal component
is controlled by the secondary constraint itself, which by construction
imposes the NC-deformed Gauss law. The spatial comparison, by contrast,
requires the full reduced Hamilton equations for $\dot\pi^s$, whose source
sector inherits both the Banerjee-map choice and the canonical deformation
of $\pi^i$. Within the Banerjee realization studied here, it is precisely
that map dependence---in particular the $F_{\alpha\beta}$-dependent terms
absent in the minimal Adorno map---that produces the residual
gauge-invariant obstruction in the Amp\`ere sector. The present analysis
therefore supports a limited conclusion: for this specific first-order
action-level current map, the reduced spatial sector does not reproduce the
gauge-covariant target equations of Ref.~\cite{Adorno2011CNE} for generic
source profiles. No broader no-go statement is claimed for other SW current
maps or for other action-level realizations not analyzed here.

\subsection{Comparison with prior external-source analyses}
\label{sec:Discussion-PriorWork}

A related structural remark concerns the role of a temporal source
component. The Cabo--Shabad construction is formulated under the
restriction $J_4^{\,a}=0$, which prevents the temporal source density from
shifting Gauss's law at the stage where their source-dependent secondary
condition emerges. Had such a component been retained, the corresponding
consistency equation would no longer remain purely spatial. The present
SW-mapped model shows explicitly how that source-compatibility content can
instead reappear at a later stage of the chain, through the tertiary-stage
candidate $\Phi_3^{\rm cand}$.

The source-dependent comparison with earlier literature should therefore
be read with some care. Cabo--Shabad remain the closest canonical analogue,
because in both analyses a prescribed external current enters the
consistency chain itself; the difference is that their non-Abelian
construction is written with $J_4^{\,a}=0$, so the source-dependent
condition appears already at the secondary stage, whereas here the temporal
density $J^0$ and the spatial current $J^k$ survive together and the
corresponding compatibility content is first isolated at the tertiary
stage~\cite{CaboShabad}. Przeszowski's analysis of Yang--Mills theory with
arbitrary external sources points in the same methodological direction:
there too the source sector cannot be treated as canonically innocuous, and
a modified source coupling is introduced before the quantization procedure
is organized~\cite{Przeszowski1988}. Sikivie and Weiss provide a different
kind of comparison, since their static external sources are used to define
physically distinct classical solution sectors rather than to study the
first-/second-class constraint split~\cite{SikivieWeiss1978}. By contrast,
Kruglov's NC Maxwell model employs a prescribed Abelian current at the
Lagrangian level, but leaves aside the Dirac--Bergmann question of how the
source deforms the canonical chain~\cite{Kruglov2002}. The gain of the
present Hamiltonian analysis is precisely to make that source dependence
explicit in a setting where the current is fixed, Abelian, SW-mapped, and
allowed to carry both temporal and spatial components.

\subsection{Scope, limitations, and outlook}
\label{sec:Discussion-Scope}

The results reported here are confined to the first-order action-level
Banerjee realization of NC $U(1)$ electrodynamics with purely space--space
non-commutativity ($\theta^{0i}=0$) and a prescribed, non-dynamical
external current $J^\mu$. Within that scope, the main identity
\eqref{eq:MainIdentityDivEL} is algebraic: it holds as a statement in the
$\mathcal{O}(\theta)$ canonical variables, without additional assumptions
on the source profile. The generic closure by multiplier fixing
(Sec.~\ref{sec:Phi3CandEvolution}) is itself a structural result: it
establishes that, for the Banerjee realization with $\Theta^l\neq0$, the
Dirac algorithm terminates without producing an independent quaternary
constraint, and that the obstruction is absorbed into the primary multiplier
along the differential operator $\Theta^l\partial_l$. This is not a
limitation of the analysis but a positive characterization of how the
unrestricted external-current theory closes canonically. The sectorial
classification of Sec.~\ref{sec:SectorStructure} and the reduced
phase-space construction of Sec.~\ref{sec:ReducedDescription} are, by
contrast, restricted: the gauge-generator, Dirac-bracket and
Hamilton-equations sector-level statements rely on the sufficient conditions
\eqref{eq:FirstClassSectorConditions}, and should not be transferred to
generic source profiles without further analysis. Higher-order extensions
in $\theta$, temporal non-commutativity $\theta^{0i}\neq0$, and dynamical
(rather than prescribed) sources each raise distinct conceptual issues that
are not addressed here.

Because the central identity is map-dependent (see the remark following
Proposition~\ref{prop:DivELIdentity}), two companion analyses are currently
under investigation. The first extends the present Hamiltonian construction
to the minimal current map of Ref.~\cite{Adorno2011CNE}, in order to
determine whether the residual gauge-invariant obstruction in the Amp\`ere
sector is specific to the Banerjee $F_{\alpha\beta}$-dependent terms or
whether some weaker form of it survives in any action-level SW current map.
The second constructs the canonical counterpart of the equation-level route
itself, \emph{i.e.}\ asks how the gauge-covariant Adorno field equations fit
in a Dirac--Bergmann analysis when the SW map is applied after variation; a
comparison between the two canonical constructions would make the relation
between the two routes quantitative at the phase-space level. The approach
of Refs.~\cite{Brace2001,Barnich2004} to SW-map ambiguities through local
cohomology, together with the current-algebra perspective of
Ref.~\cite{BanerjeeGhosh2002}, suggests natural organizing principles for
such extensions. Neither task is required for the present claim, but
together they would delimit how much of the obstruction identified here is
tied to the Banerjee realization and how much is structural.

\section{Conclusions}
\label{sec:Conclusions}
%==============================================================================

We studied the first-order Banerjee action-level realization of NC Maxwell
theory with a fixed external current in the full Dirac--Bergmann phase space.
The main result is that preservation of the Gauss-type secondary constraint
produces a third-stage source-dependent object whose canonical form is
algebraically identical, after canonical pullback, to the divergence of the
mapped Euler--Lagrange equations. This identity is tied to the specific
$F_{\alpha\beta}$-dependent terms of the Banerjee current map: it
characterizes the canonical fingerprint of that particular realization, not a
universal feature of any action-level SW-mapped theory. The Hamiltonian
analysis thereby provides three layers of information beyond the Lagrangian
diagnosis of Ref.~\cite{Adorno2011CNE}: it localizes the obstruction at the
tertiary stage of the Dirac chain, classifies the source dependence through
the hierarchical kernel structure
$(\Theta^l,\Sigma^{lk},\Xi^l)$, and decomposes the on-shell
comparison into a Gauss sector recovered exactly and an Amp\`ere sector with
irreducible residues.

For generic source profiles with $\Theta^l\neq0$, the tertiary consistency
condition feeds back into the primary multiplier and the algorithm closes by
multiplier determination rather than by a quaternary constraint. This
generic closure is itself a structural result: it establishes how the
unrestricted external-current theory terminates canonically within the
Banerjee realization. A reduced Hamiltonian description is recovered only
after restricting to a sufficient first-class sector in which the displayed
source-dependent kernels and the tertiary object become trivial. The gauge
generator, reduced brackets, and degree-of-freedom count established in the
manuscript are therefore strictly sectorial.

The scope of the paper is correspondingly limited but precise. It does not
provide a complete first-/second-class classification for arbitrary external
currents, nor a general reduced phase-space formulation of the full theory.
Rather, it isolates the source-dependent obstruction in the Dirac chain,
clarifies its map-dependent relation to the Euler--Lagrange structure, and
identifies the restricted sector in which the standard equal-time
Hamiltonian machinery remains available. The extension of this analysis to
the minimal current map of Ref.~\cite{Adorno2011CNE} and to the canonical
counterpart of the equation-level route are currently under
investigation.

%==============================================================================
\section*{Acknowledgements}
The authors welcome the support of the Universidad Juárez Autónoma de Tabasco for providing a suitable work environment while this research was carried out. J.M.C. also thanks the Secretariat of Science, Humanities, Technology, and Innovation (Secihti) of México for their support through a grant for postdoctoral studies under Grant No. 3873825. A.G.A.C. acknowledges support from the National Council of Humanities, Science and Technology (CONAHCYT), through a master's scholarship awarded as part of the national postgraduate funding program.

%==============================================================================

\section*{Statements and Declarations}

\textbf{Competing interests.} The authors declare that they have no competing
financial or non-financial interests that are directly or indirectly related
to the work submitted for publication.

\textbf{Author contributions.} J.M.C. conceived the study, performed the
Dirac--Bergmann analysis and the on-shell comparison, and wrote the
manuscript. A.G.A.C. developed and implemented the symbolic-algebra toolbox
used to verify the displayed formulas, performed computational verification
of the results, and contributed to the manuscript preparation. J.M.P.F.
supervised the project, contributed to the structural analysis of the
constraint chain, and participated in the revision of the manuscript. All
authors reviewed and approved the final version.

%\textbf{Use of AI tools.} Portions of the manuscript text were improved for clarity, style and grammar using an AI-assisted copy-editing tool (Claude, Anthropic). No AI tool was used for autonomous content creation or for generating new scientific results. All mathematical derivations, physical interpretations and conclusions are the sole work of the authors, who takefull responsibility for the content of the manuscript.

%==============================================================================
\section*{Data and Code Availability Statement}
The results reported here are obtained from analytic and symbolic manipulations; no new experimental data were generated. The Matlab symbolic-algebra toolbox used for the constrained Hamiltonian (Dirac--Bergmann) analysis is available at \href{https://github.com/AlejandroAndarcia/MATLAB-Implementation-of-the-Dirac-Bergmann-Algorithm-for-Field-Theories}{https://github.com/AlejandroAndarcia/MATLAB-Implementation-of-the-Dirac-Bergmann-Algorithm-for-Field-Theories} under the license specified in the repository. Documentation, usage examples, and system requirements are provided there. In the present work the code is used only to verify the analytic formulas displayed in the manuscript; no claim depends exclusively on a hidden computational step.

%==============================================================================
\appendix
%==============================================================================

%==============================================================================
\section{Functional identities for the closure calculation}
\label{app:PBclosure}
%==============================================================================

This appendix collects the functional-derivative identities used in the closure calculation. We only need the contracted combinations that enter the Poisson brackets
$\{\phi_2,H_{\text{ws}}\}$ and $\{\phi_2,H_J\}$. The basic functional derivatives
of the secondary constraint are
\begin{align}
\frac{\delta \phi_2(x)}{\delta A_k(y)}
&=
\frac12\,\theta^{kl}\,\partial_l^x
\!\bigl(J^0(x)\,\delta^3(x-y)\bigr),
\label{eq:AppdphiA}
\\
\frac{\delta \phi_2(x)}{\delta \pi^k(y)}
&=
\partial_k^x\,\delta^3(x-y).
\label{eq:Appdphipi}
\end{align}
Inserting these into the canonical field-theory Poisson bracket gives
\begin{equation}
\{\phi_2,H\}
=
\int d^3y\,
\left[
\frac{\delta \phi_2}{\delta A_k(y)}\frac{\delta H}{\delta \pi^k(y)}
-
\frac{\delta \phi_2}{\delta \pi^k(y)}\frac{\delta H}{\delta A_k(y)}
\right].
\label{eq:AppPBFormula}
\end{equation}

For the source-free Hamiltonian, only the following contracted combination is
needed:
\begin{equation}
\begin{aligned}
\{\phi_2,H_{\text{ws}}\}
{}&=
\frac12\,\theta^{kl}(\partial_lJ^0)\,\partial_kA_0
\\ &-\frac12\,\theta^{kl}\eta_{00}\eta_{mk}\partial_l(J^0\pi^m),
\label{eq:AppPBws}
\end{aligned}
\end{equation}
which coincides with Eq.~\eqref{eq:PBws_result} in the main text. The result
arises from the two terms that generate non-zero contributions at
$\mathcal{O}(\theta)$: the linear-in-$\pi$ term $H_A=\pi^k\partial_kA_0$
(via the operator $\{\phi_2^{(1)},H_A\}$) and the quadratic-in-$\pi$ term
$H_D=-\tfrac12\eta_{00}\eta_{ij}\pi^i\pi^j$ (via $\{\phi_2^{(1)},H_D\}$).
All remaining terms of $H_{\text{ws}}$ contribute zero, either by the Ward gauge
identity (terms depending only on $F_{ij}$), by the antisymmetry of
$\theta^{kl}$ against symmetric second derivatives, or because they are
$\mathcal{O}(\theta^2)$.

For the source Hamiltonian, the contraction with
Eqs.~\eqref{eq:AppdphiA}--\eqref{eq:Appdphipi} yields
\begin{equation}
\begin{aligned}
\{\phi_2,H_J\}
={}& -\partial_i J^i
-\tfrac{1}{2}\theta^{il}(\partial_iJ^0)\,\partial_lA_0
\\ &
+ \tfrac{1}{2}\theta^{jl}\partial_i\partial_l(J^iA_j)
- \tfrac{1}{2}\theta^{il}\partial_i(J^k\partial_lA_k)
\\
&
-\tfrac{1}{2}\theta^{il}\eta_{00}\eta_{lk}\partial_i(J^0\pi^k)
+ \tfrac{1}{2}\theta^{il}\partial_i(J^kF_{kl}),
\end{aligned}
\label{eq:AppPBJ}
\end{equation}
which coincides with Eq.~\eqref{eq:PBJ_result} of the main text.
Each line originates from a specific subset of terms in $\mathcal H_J$:
the first term from $A_jJ^j$; the second from $A_jJ^0\partial_lA_0$;
the third from $A_jJ^k\partial_lA_k$ (two contributions combined);
the fourth from $A_jJ^0\eta_{00}\eta_{li}\pi^i$; and the fifth from
$A_jF_{kl}J^k$. The $\pi$-containing terms from $\{\phi_2,H_{\text{ws}}\}$
(Eq.~\eqref{eq:AppPBws}) and from the fourth term here cancel exactly
against each other in the sum $\dot\phi_2$, so that
Eq.~\eqref{eq:PBintermediate} contains no $\pi^k$ at all.

Combining Eqs.~\eqref{eq:AppPBws} and \eqref{eq:AppPBJ} with the explicit
derivative \eqref{eq:PhiExpl} reproduces Eq.~\eqref{eq:PBintermediate}.
The structure of that expression has three key properties. First, all
$\pi^k$-containing terms cancel algebraically between $\{\phi_2,H_{\text{ws}}\}$
and $\{\phi_2,H_J\}$, so the constraint surface $\Gamma_c=\{\phi_1\approx0,
\phi_2\approx0\}$ has no effect on $\dot\phi_2$. Second, the $\mathcal
O(\theta^0)$ piece is $-\partial_\mu J^\mu$: a condition on the external
source, not on the phase space, and the direct analogue of the Maxwell
consistency condition. Third, the $\mathcal{O}(\theta)$ piece contains
$A_\mu$ and $\partial_0J^0$, so the full tertiary candidate depends on the
canonical coordinates but not on the canonical momenta. Accordingly,
$\{\Phi_3^{\rm cand},\Phi_3^{\rm cand}\}=0$ whereas the primary--tertiary
bracket remains
$\{\phi_1,\Phi_3^{\rm cand}\}=-\Theta^l\partial_l\delta$ in the
inhomogeneous-source sector. The next Dirac step therefore feeds back into the
primary multiplier through the preservation of $\Phi_3^{\rm cand}$, instead of
being settled by its self-bracket alone.

%==============================================================================
\section{Constraint brackets through the tertiary stage}
\label{app:SelfBrackets}
%==============================================================================

This appendix provides the explicit derivations of $\{\phi_2,\phi_2\}$
and $\{\phi_2,\Phi_3^{\rm cand}\}$. These kernels enter the sector discussion
of Sec.~\ref{sec:SectorStructure}.

\subsection*{Secondary self-bracket $\{\phi_2(x),\phi_2(y)\}$}

Using the canonical Poisson bracket in functional form,
\begin{equation}
\{F,G\}
=
\int d^3z\,
\left[
\frac{\delta F}{\delta A_k(z)}\frac{\delta G}{\delta \pi^k(z)}
-
\frac{\delta F}{\delta \pi^k(z)}\frac{\delta G}{\delta A_k(z)}
\right],
\label{eq:AppPBFormulaSelf}
\end{equation}
and inserting the functional derivatives of $\phi_2$ from
Eqs.~\eqref{eq:AppdphiA}--\eqref{eq:Appdphipi}, one obtains
\begin{align}
\{\phi_2(x),\phi_2(y)\}
={}&
\int d^3z\,
\Bigl[
\tfrac{1}{2}\theta^{kl}\partial_l^x\!\bigl(J^0(x)\delta^3(x-z)\bigr)\,
\times \nonumber\\
&\partial_k^z\delta^3(y-z)-\,
\partial_k^x\delta^3(x-z)\,
\times \nonumber\\
&\tfrac{1}{2}\theta^{kl}\partial_l^z\!\bigl(J^0(y)\delta^3(y-z)\bigr)
\Bigr].
\label{eq:AppSelfBracketIntegral}
\end{align}
Evaluating the first term: since $J^0(x)$ is independent of $z$, the
factor $\partial_l^x(J^0(x)\delta^3(x-z))=(\partial_l J^0(x))\delta^3(x-z)$
when the $x$-derivative acts only on $J^0$, plus a total-derivative term
that integrates to zero under suitable boundary conditions. After integrating
over $z$ using $\delta^3(x-z)$,
\begin{equation}
\begin{aligned}
\text{(first term)}
{}&= \tfrac{1}{2}\theta^{kl}(\partial_l J^0(x))\,\partial_k^y\delta^3(y-x)
\\ & =
-\tfrac{1}{2}\Theta^k(x)\,\partial_k^x\delta^3(x-y)\,.
\label{eq:AppSelf1}
\end{aligned}
\end{equation}
Evaluating the second term analogously, with the $z$-integration setting
$z=y$,
\begin{equation}
\begin{aligned}
\text{(second term)}
{}&=
-\tfrac{1}{2}\theta^{kl}(\partial_l J^0(y))\,\partial_k^x\delta^3(x-y)
\\ &=
\tfrac{1}{2}\Theta^k(y)\,\partial_k^y\delta^3(x-y)\,.
\label{eq:AppSelf2}
\end{aligned}
\end{equation}
Adding Eqs.~\eqref{eq:AppSelf1} and \eqref{eq:AppSelf2} and using
$\partial_k^x\delta^3(x-y)=-\partial_k^y\delta^3(x-y)$,
\begin{equation}
\begin{aligned}
\{\phi_2(x),\phi_2(y)\}
{}&=
\tfrac{1}{2}\Theta^k(y)\,\partial_k^y\delta^3(x-y)
\\&-\tfrac{1}{2}\Theta^k(x)\,\partial_k^x\delta^3(x-y)\,,
\label{eq:AppPhi2Phi2Final}
\end{aligned}
\end{equation}
which is manifestly antisymmetric and gives the secondary self-bracket used in Sec.~\ref{sec:SectorStructure}
of the main text. This bracket vanishes if and only if $\Theta^k=0$ everywhere.

\subsection*{Secondary--tertiary bracket $\{\phi_2(x),\Phi_3^{\rm cand}(y)\}$}

Write the compact first-order tertiary candidate as
\begin{equation}
\begin{aligned}
\Phi_3^{\rm cand}
{}&=
-C
+\Theta^l\partial_lA_0
+\theta^{li}(\partial_iJ^k)\partial_lA_k
\\ &-\frac{1}{2}\theta^{il}\partial_i\!\bigl(A_l C\bigr), \qquad
C:=\partial_\mu J^\mu,
\label{eq:AppPhi3CandCompact}
\end{aligned}
\end{equation}
where we used the antisymmetry of $\theta^{li}$ to simplify
$\theta^{li}\partial_i(J^k\partial_lA_k)=\theta^{li}(\partial_iJ^k)\partial_lA_k$.
Since $\Phi_3^{\rm cand}$ contains no canonical momenta and only the
$\partial_m\pi^m$ piece of $\phi_2$ has non-vanishing brackets with the
spatial gauge potential, one immediately has
\begin{equation}
\begin{aligned}
\{\partial_m\pi^m(x),A_k(y)\}&=\partial_k^y\delta^3(x-y),
\\
\{\partial_m\pi^m(x),A_0(y)\}&=0.
\label{eq:AppBasicPhi2Phi3CandBrackets}
\end{aligned}
\end{equation}
Therefore,
\begin{align}
\{\phi_2(x),\Phi_3^{\rm cand}(y)\}
{}&=\,\theta^{li}(\partial_iJ^k)(y)\,\partial_l^y\partial_k^y\delta^3(x-y)
\nonumber\\ &
-\frac{1}{2}\theta^{il}\partial_i^y\!\Bigl(C(y)\,\partial_l^y\delta^3(x-y)\Bigr).
\label{eq:AppPhi2Phi3CandStep1}
\end{align}
Expanding the last term,
\begin{equation}
\theta^{il}\partial_i\!\Bigl(C\,\partial_l\delta^3\Bigr)
=
\theta^{il}(\partial_i C)\partial_l\delta^3
+\theta^{il}C\,\partial_i\partial_l\delta^3.
\label{eq:AppPhi2Phi3CandExpansion}
\end{equation}
The second contribution vanishes because $\partial_i\partial_l\delta^3$ is
symmetric in $(i,l)$ whereas $\theta^{il}$ is antisymmetric. One is left with
\begin{equation}
\begin{aligned}
\{\phi_2(x),\Phi_3^{\rm cand}(y)\}
{}&=
\theta^{li}(\partial_iJ^k)(y)\,\partial_l^y\partial_k^y\delta^3(x-y) \\ &
-\frac{1}{2}\theta^{il}(\partial_i C)(y)\,\partial_l^y\delta^3(x-y),
\label{eq:AppPhi2Phi3CandFinal}
\end{aligned}
\end{equation}
which is the mixed kernel quoted in Sec.~\ref{sec:SectorStructure}. The first
term depends on spatial current gradients, while the second is driven by
gradients of the full divergence $\partial_\mu J^\mu$.

\subsection*{Kernel decomposition of $\{\phi_2(x),\Phi_3^{\rm cand}(y)\}$}

Equation~\eqref{eq:AppPhi2Phi3CandFinal} may be decomposed into the symmetric
spatial-current kernel and the divergence-gradient kernel introduced in the
main text. Defining
\begin{equation}
\Sigma^{lk}
:=
\frac{1}{2}\Bigl(\theta^{li}\partial_iJ^k+\theta^{ki}\partial_iJ^l\Bigr),
\label{eq:AppSigmaDef}
\end{equation}
and
\begin{equation}
\Xi^l
:=
\theta^{il}\partial_i C,
\label{eq:AppXiDef}
\end{equation}
one obtains
\begin{equation}
\begin{aligned}
\{\phi_2(x),\Phi_3^{\rm cand}(y)\}
{}&=
\Sigma^{lk}(y)\,\partial_l^y\partial_k^y\delta^3(x-y)
\\ &-\frac{1}{2}\Xi^l(y)\,\partial_l^y\delta^3(x-y),
\label{eq:AppPhi2Phi3CandDecomposed}
\end{aligned}
\end{equation}
because only the symmetric part of the coefficient of
$\partial_l\partial_k\delta^3(x-y)$ contributes. This is the form used in
Sec.~\ref{sec:SectorStructure} to state the explicit kernel conditions for the
restricted sufficient first-class subcase.

%==============================================================================
\subsection*{First-order form of $\dot{\Phi}_3^{\rm cand}$}

This appendix records the explicit first-order structure of the next Dirac step starting from the tertiary-stage candidate $\Phi_3^{\rm cand}=\dot\phi_2$. Since the candidate contains no momentum variables, its self-bracket vanishes identically. The only bracket relevant for multiplier fixing is the one with the primary constraint,
\begin{equation}
\begin{aligned}
\{\Phi_3^{\rm cand}(x),\phi_1(y)\}
&=
\Theta^l(x)\partial_l^x\delta^3(x-y), \\
\Theta^l&:=\theta^{li}\partial_iJ^0.
\end{aligned}
\end{equation}
Retaining terms through first order in $\theta^{ij}$, one finds the schematic preservation equation quoted in the main text,
\begin{align}
\dot\Phi_3^{\rm cand}
={}&
\Theta^l \partial_l u_1
+\theta^{li}(\partial_i J^m)\partial_l \partial_m A_0 \nonumber\\
&
-\eta_{00}\theta^{li}(\partial_i J_m)\partial_l \pi^m 
+\theta^{li}(\partial_i\partial_0 J^0)\partial_l A_0 \nonumber\\
&
+\theta^{li}(\partial_i\partial_0 J^k)\partial_l A_k
-\partial_0(\partial_\mu J^\mu),
\label{eq:Phi3CandEvolutionAppendix}
\end{align}
where the last term comes from the explicit time dependence of the source-only part of $\Phi_3^{\rm cand}$. The first term is the only one involving the arbitrary multiplier. Therefore, in regions where $\Theta^l\neq0$, the condition $\dot\Phi_3^{\rm cand}\approx0$ constrains $u_1$ along the vector field $\Theta^l\partial_l$ and closes the generic unrestricted branch by multiplier fixing. Only in degenerate sectors, such as $\Theta^l=0$ together with additional cancellations, can the algorithm continue beyond this stage.

The formulas used in the reduced description of
Sec.~\ref{sec:ReducedDescription} refer only to the further restricted
sufficient first-class subcase in which the full tertiary condition
$\Phi_3^{\rm cand}$ vanishes identically and the equal-time kernels listed
in Sec.~\ref{sec:SectorStructure} disappear.

%==============================================================================

%==============================================================================
\section{Independent proof of the divergence identity}
\label{app:ELDivergenceIdentity}
%==============================================================================

This appendix gives an independent derivation of the identity
\begin{equation}
\partial_\nu\mathcal E_B^\nu
=\dot\phi_2,
\label{eq:AppMainIdentity}
\end{equation}
using only the primary and secondary constraints, the decomposition \eqref{eq:PhiDotSplit}, and the explicit derivative \eqref{eq:PhiExpl}. No simplified formula from the main text is assumed.

\subsection*{Divergence of the Euler--Lagrange equations}

Write the Banerjee-mapped Lagrangian as $\mathcal L=\mathcal L_{\text{ws}}(F)+\mathcal L_J$, where
\begin{equation}
\mathcal{L}_J
=
-A_\mu J^\mu
-\theta^{\mu\alpha}J^\beta A_\mu\partial_\alpha A_\beta
+\frac12\theta^{\mu\alpha}J^\beta A_\mu\partial_\beta A_\alpha.
\end{equation}
Since $\mathcal L_{\text{ws}}$ depends only on the antisymmetric field strength, its Euler--Lagrange contribution has identically vanishing divergence. It is therefore enough to vary $\mathcal L_J$.

A direct application of the Euler--Lagrange operator gives
\begin{equation}
\begin{aligned}
\mathcal E_J^\nu
={}& -J^\nu
-\theta^{\nu\alpha}J^\beta F_{\alpha\beta}
-\theta^{\alpha\beta}(\partial_\alpha J^\nu)A_\beta \\
&\quad
-\theta^{\alpha\beta}J^\nu\partial_\alpha A_\beta
+\tfrac12\theta^{\nu\beta}(\partial_\mu J^\mu)A_\beta.
\end{aligned}
\label{eq:AppELJIndependent}
\end{equation}
Taking its divergence and using $\theta^{0i}=0$ yields
\begin{align}
\partial_\nu\mathcal E_B^\nu
={}&
-\partial_\mu J^\mu
-\partial_i(\theta^{il}J^\beta F_{l\beta})
-\partial_\nu\!\bigl(\theta^{il}(\partial_iJ^\nu)A_l\bigr)
\nonumber\\
&
-\partial_\nu\!\bigl(\theta^{il}J^\nu\partial_iA_l\bigr)
+\frac12\partial_i\!\bigl(\theta^{il}(\partial_\mu J^\mu)A_l\bigr).
\end{align}
Expanding the time and space derivatives separately, using the antisymmetry of $\theta^{il}$ and the commutativity of partial derivatives, one finds the compact result
\begin{equation}
\begin{aligned}
\partial_\nu \mathcal E_B^\nu
={}& -\partial_\mu J^\mu
+\theta^{li}(\partial_iJ^0)\partial_lA_0
+\theta^{li}\partial_i(J^k\partial_lA_k) \\
&\quad
-\frac12\theta^{il}\partial_i\!\bigl[A_l(\partial_\mu J^\mu)\bigr].
\end{aligned}
\label{eq:AppDivELCompact}
\end{equation}

\subsection*{Independent computation of $\dot\phi_2$}

Starting from the accepted secondary constraint
\begin{equation}
\phi_2
=\partial_i\pi^i-J^0+\frac12\theta^{kl}\partial_l(A_kJ^0),
\end{equation}
one has the functional derivatives
\begin{align}
\frac{\delta \phi_2(x)}{\delta A_k(y)}
&=
\frac12\theta^{kl}\partial_l^x\!\bigl(J^0(x)\delta^3(x-y)\bigr),
\\
\frac{\delta \phi_2(x)}{\delta \pi^k(y)}
&=
\partial_k^x\delta^3(x-y).
\end{align}
Inserting them into the functional Poisson-bracket formula and evaluating the contributions of $H_{\text{ws}}$ and $H_J$ independently, one obtains
\begin{align}
\{\phi_2,H_{\text{ws}}\}
{}&=\frac12\theta^{kl}(\partial_lJ^0)\partial_kA_0\nonumber\\
&
-\frac12\theta^{jl}\eta_{00}\eta_{jn}\partial_l(J^0\pi^n),
\\
\{\phi_2,H_J\}
={}&-\partial_iJ^i
-\tfrac12\theta^{il}\partial_i(J^0\partial_lA_0) \nonumber\\
&
+\tfrac12\theta^{jl}\partial_i\partial_l(J^iA_j)- \tfrac12\theta^{il}\partial_i(J^k\partial_lA_k) \nonumber\\
&
-\tfrac12\theta^{il}\eta_{00}\eta_{lk}\partial_i(J^0\pi^k)
+\tfrac12\theta^{il}\partial_i(J^kF_{kl}).
\end{align}
The momentum terms cancel exactly in the sum. Adding also the explicit derivative \eqref{eq:PhiExpl}, one recovers
\begin{align}
\dot\phi_2
={}&-\partial_\mu J^\mu
+\theta^{li}(\partial_iJ^0)\partial_lA_0
+\tfrac12\theta^{jl}\partial_i\partial_l(J^iA_j)
\nonumber\\
&+\theta^{li}\partial_i(J^k\partial_lA_k)
+\tfrac12\theta^{il}\partial_i(J^k\partial_kA_l)\nonumber\\
&
+\tfrac12\theta^{kl}\partial_l(A_k\partial_0J^0).
\end{align}
The last three terms may be rearranged algebraically into
\begin{equation}
-\frac12\theta^{il}\partial_i\!\bigl[A_l(\partial_\mu J^\mu)\bigr].
\end{equation}
so that
\begin{equation}
\begin{aligned}
\dot\phi_2
={}&
-\partial_\mu J^\mu
+\theta^{li}(\partial_iJ^0)\partial_lA_0
+\theta^{li}\partial_i(J^k\partial_lA_k)
\nonumber\\
& -\frac12\theta^{il}\partial_i\!\bigl[A_l(\partial_\mu J^\mu)\bigr].
\label{eq:AppPhi2DotCompact}
\end{aligned}
\end{equation}
Comparing Eqs.~\eqref{eq:AppDivELCompact} and \eqref{eq:AppPhi2DotCompact} proves Eq.~\eqref{eq:AppMainIdentity}. In the unrestricted source problem, the third Dirac-stage candidate is therefore algebraically identical to the canonical pullback of the divergence of the Banerjee Euler--Lagrange equations.

\end{document}